\documentclass[11pt]{article}
\usepackage{authblk, amsmath, openbib, epsf, graphicx}
\usepackage{natbib}
\bibpunct[, ]{(}{)}{;}{a}{}{,}
\usepackage[pdftex,colorlinks=true,linkcolor=blue,citecolor=blue,urlcolor=blue]{hyperref}
\usepackage{doi}
\usepackage{amsmath,amssymb,amsthm,bbm}
\usepackage{mathrsfs}
\usepackage{booktabs,dcolumn}
\usepackage{caption}
\usepackage{subcaption}
\usepackage{threeparttable}
\usepackage{booktabs}
\usepackage{rotating,booktabs}
\usepackage[normalem]{ulem}  
\usepackage{indentfirst}
\usepackage{array}
\usepackage{times}
\usepackage[usenames]{color}
\usepackage{psfrag,graphicx}
\usepackage{graphicx,stackengine}
\usepackage{eso-pic,graphicx}
\usepackage{graphicx}
\usepackage{type1cm}
\usepackage{bm}
\usepackage{float}
\usepackage{multirow}
\usepackage{silence}
\usepackage{color}
\usepackage{rotate}
\AtBeginEnvironment{quote}{\itshape}

\newcommand\underlay[4]{%
  \stackengine{0pt}%
  {\kern#2\includegraphics[height=#1]{#4}}%
  {\includegraphics[height=#1]{#3}}%
  {O}{l}{F}{F}{L}%
}

\newcommand{\utwi}[1]{\boldsymbol{#1}}

\newcolumntype{z}[1]{D{.}{.}{#1}}

\newcommand{\ignore}[1]{}{}
\newcommand{\R}{\ensuremath{\mathbb{R}}}
\newcommand{\N}{\ensuremath{\mathbb{N}}}

\newcommand{\bone}{\mathbf{1}}
\newcommand{\bzero}{\mathbf{0}}
\newcommand{\id}{\mathbbm{1}}
\renewcommand{\S}{\operatorname{S}}
\newcommand{\V}{\operatorname{V}}
\newcommand{\bx}{\bm{x}}
\newcommand{\bX}{\bm{X}}

\date{}
\usepackage{comment}
\theoremstyle{definition}
\newtheorem{proposition}{Proposition}

\newtheorem{remark}{Remark}

\begin{document}
\begin{titlepage}
\title{Forecasting and Backtesting Gradient Allocations of Expected Shortfall}

\author{
Takaaki Koike\thanks{Corresponding author: Takaaki Koike, E-mail: takaaki.koike@r.hit-u.ac.jp }$\:\:^{1}$,
Cathy W.S. Chen$^{2}$, and Edward M.H. Lin$^{3}$\\
$^{1}$Graduate School of Economics, Hitotsubashi University, Japan\\
$^{2}$Department of Statistics, Feng Chia University, Taiwan\\
$^{3}$Department of Statistics, Tunghai University, Taiwan
}

\maketitle
\setcounter{page}{1}

\begin{abstract}
Capital allocation is a procedure for quantifying the contribution of each source of risk to aggregated risk.
The gradient allocation rule, also known as the Euler principle, is a prevalent rule of capital allocation under which the allocated capital captures the diversification benefit of the marginal risk as a component of overall risk.
This research concentrates on Expected Shortfall (ES) as a regulatory standard and focuses on the gradient allocations of ES, also called ES contributions (ESCs).
We present the comprehensive treatment of backtesting the tuple of ESCs 
in the framework of the traditional and comparative backtests based on the concepts of joint identifiability and multi-objective elicitability.
For robust forecast evaluation against the choice of scoring function, we also extend the Murphy diagram, a graphical tool to check whether one forecast dominates another under a class of scoring functions, to the case of ESCs.
Finally, leveraging the recent concept of multi-objective elicitability, we propose a novel semiparametric model for forecasting dynamic ESCs based on a compositional regression model.
In an empirical analysis of stock returns we evaluate and compare a variety of models for forecasting dynamic ESCs and demonstrate the outstanding performance of the proposed model.
\end{abstract}


\end{titlepage}

\newpage

\section{Introduction}\label{sec:intro}

Risk aggregation and risk allocation are two important procedures for financial risk management.
The total risk of a portfolio or a financial institution is quantified in the first stage of risk aggregation, where financial regulation may require the use of certain risk measures, such as \emph{value-at-risk (VaR)} and \emph{expected shortfall (ES)}.
Due to some deficiencies of VaR,
ES has become a standard risk measure in banking regulation~\citep{bcbs2016minimum,bcbs2019minimum}.
Since ES is coherent and particularly subadditive, it provides an incentive to diversify the risks in a financial institution or a portfolio.
On the other hand, this is not the case for VaR, which discourages its use in risk aggregation.
Readers may refer to~\citet{emmer2015best} and~\citet{mcneil2015quantitative} for a comprehensive discussion on VaR and ES in financial risk management.

Capital allocation is a second stage for quantifying the contribution of each source of risk to total risk.
Reflecting the diversification benefit, the allocated capital to each component is calculated so that the sum of the allocated capital over all components equals the total capital.
Among various allocation rules proposed in the literature, the \emph{Euler principle} is a prevalent allocation rule justified from various perspectives, such as RORAC (return on risk-adjusted capital) compatibility~\citep{tasche1999risk}, cooperative game theory~\citep{denault2001coherent}, and axiomatic foundation of capital allocations~\citep{kalkbrener2005axiomatic}; see also~\citet{tasche2008capital} and references therein.
Together with the growing importance of ES in recent banking regulations, this paper focuses on the situation of capital allocation when total risk is measured by ES and allocated capital is calculated under the Euler principle; see Section~\ref{sec:allocation:ES}
 for details. We call the allocated capital calculated in this setting \emph{ES contribution (ESC)}.

The rare event nature of the above tail risk quantities poses challenges in their statistical estimation and model evaluation.
Forecast evaluation of tail risk quantities is called \emph{backtesting} in financial terminology.
Regarding these issues, \emph{elicitability} and \emph{identifiability} are important concepts of risk functionals, where elicitability allows for a comparative study of forecasting models, and identifiability offers an absolute criterion of forecast accuracy.
These notions are also beneficial for estimating and calibrating models of tail risk quantities; see~\citet{nolde2017elicitability}.
It is known that VaR is elicitable and identifiable by itself, whereas ES is elicitable and identifiable when it is simultaneously estimated with VaR~\citep{fissler2016higher}.
These results open the door to a semiparametric estimation of risk measures based on scoring functions; see, for example,~\citet{patton2019dynamic} and~\citet{taylor2019forecasting}.

In contrast to the case of risk measures, studies on backtesting of dynamic capital allocation are quite limited in the literature.
\citet{bielecki2020fair} propose a backtesting method of ESCs based on the concept called fairness.
This method is essentially treated in the more general framework of calibration and traditional backtests introduced in~\citet{nolde2017elicitability}.
Despite the usefulness of such traditional backtests for model validation, \citet{nolde2017elicitability} also point out several issues of such tests particularly in terms of model comparison and banking regulation.
To this end, our first aim in this paper is to present the comprehensive treatment of backtesting ESCs 
in the framework of traditional and comparative backtests~\citep{nolde2017elicitability} based on the concepts of \emph{joint identifiability} and \emph{multi-objective elicitability}~\citep{fissler2024backtesting}.

Although the results shown in~\citet{fissler2024backtesting} are on systemic risk measures such as CoVaR, CoES and MES, in Section~\ref{sec:elicitability:identifiability} we translate their results under the setting of capital allocation.
Summarizing their results in our setting, ESCs are jointly identifiable with total VaR (i.e., VaR of the aggregated loss); moreover, ESCs themselves are not elicitable, but are multi-objective elicitable, when combined with total VaR and with respect to the lexicographical order.
The last statement means more precisely that the tuple of true ESCs and total VaR is elicited as a unique minimizer of an expectation of some $\mathbb{R}^2$-valued scoring function with respect to the order on $\mathbb{R}^2$ such that the forecast evaluation of VaR is prioritized, and that of ESCs
is conducted secondarily.
Therefore, in the framework of a comparative backtest based on multi-objective elicitability it is necessary to forecast ESCs together with total VaR, and two forecasts of ESCs are comparable only when the corresponding forecasts of total VaR are equally accurate.
Section~\ref{sec:backtesting:dynamic:ES:cont} summarizes the framework of backtests of dynamic ESCs.

When conducting a comparative backtest in practice, it may not always be clear which scoring function to use.
In addition, the resulting forecast ranking can be sensitive to the choice of scoring functions~\citep{patton2020comparing}.
To overcome these problems, a diagnostic tool called the \emph{Murphy diagram} is explored in~\citet{ehm2016quantiles} and~\citet{ziegel2020robust} for robust forecast evaluation of VaR and ES, respectively, against the choice of scoring functions.
Based on a mixture representation of a scoring function, the Murphy diagram is beneficial for checking whether one forecast dominates another under a class of scoring functions.

As a second aim of this paper,  we complement the framework of backtesting ESCs by developing their Murphy diagrams.
They serve as a visual, intuitive way to compare competing forecasting models, thereby making it easier to digest and make informed decisions based on their performance. 
Section~\ref{sec:robust} introduces Murphy diagrams of ESCs and summarizes their properties.

Our last aim of the paper is to leverage the concept of multi-objective elicitability and introduce a novel semiparametric model of dynamic ESCs.
Section~\ref{sec:proposed:model} proposes such a model, which combines the joint semiparametric estimation of the pair of total VaR and total ES with the \emph{compositional regression} for modeling the dynamics of the proportions of total ES to each component of risk.
Based on multi-objective elicitability of ESCs, forecast accuracy of total VaR is first evaluated prior to that of ESCs.
In addition, risk allocation is conceptually a second stage in the risk management process after quantifying overall risk.
This theoretical and practical two-stage procedure motivates us to first model the dynamics of the pair of total VaR and total ES and then estimate ESCs through the proportion of each source of risk to total ES.
Since the component-wise sum of the vector of proportions must equal $1$, we model its dynamics by the \emph{compositional regression}, which is a multiple regression for such compositional data.

Note that~\citet{boonen2019forecasting} also apply compositional data analysis to the problem of capital allocation.
They first estimate allocated capital from the losses under the assumption of normality and then fit and evaluate compositional regression models by regarding the set of normalized allocated capital as compositional data.
One potential challenge of this procedure is that statistical error and model uncertainty reside not only in the compositional regression model, but also in the estimated allocated capital.
To overcome this issue, we fit and evaluate compositional regression models based on a scoring function of ESCs.
Multi-objective elicitability of ESCs allows us to quantify the accuracy of compositional regression models based on realized losses and not on estimated allocated capital.
The major advantage of our proposed semiparametric model is in the complete separation of modeling ESCs from that of total VaR and total ES.
This feature enables us to concentrate on modeling ESCs based on the existing model for the pair of total VaR and total ES and to assess the forecast accuracy of ESCs among the proposed compositional regression models and an existing one by equalizing their total VaR and total ES.

Section~\ref{sec:empirical} conducts an empirical analysis of a portfolio of stock returns and demonstrates superior performance of the proposed model compared with other models including those with conditional heteroskedastic volatilities, hysteretic effects, and time-varying correlations.
Section~\ref{sec:conclusion:outlook} discusses potential directions for future research. We defer all technical results to Section~\ref{sec:proof:details}, where references starting with the prefix ``S'' refer to the supplementary material.
We also present details and additional results of our empirical analysis in Section~\ref{sec:details:empirical}.

\section{Backtesting of ES contributions}\label{sec:backtesting}

\subsection{Gradient allocations of expected shortfall}\label{sec:allocation:ES}

Throughout the paper 
we fix an atomless probability space $(\Omega,{\cal A},\mathbb{P})$, where all random objects are defined.
For $p \in \{0\}\cup[1,\infty)$ and $d \in \N$, let ${\cal L}^p(\R^d)$ be the set of all $\R^d$-value random vectors on $(\Omega,{\cal F},\mathbb{P})$ whose components have a finite $p$th moment.
Let $F_{\bX}$ be the joint cumulative distribution function (cdf) of $\bX \in {\cal L}^0(\R^d)$.
We also denote by ${\cal F}^p(\R^d)=\{F_{\bX}: \bX \in {\cal L}^p(\R^d)\}$ and by ${\cal F}_\text{c}^p(\mathbb{R}^d)$ the class of cdfs $F \in {\cal F}^p(\R^d)$ with a strictly positive (Lebesgue) density for every $\bx\in \mathbb{R}^d$ such that $F(\bx)\in (0,1)$.
Analogously, ${\cal L}_\text{c}^p(\mathbb{R}^d)$ denotes the set of $\mathbb{R}^d$-valued random vectors whose cdf belongs to ${\cal F}_\text{c}^p(\mathbb{R}^d)$.
For $X\in {\cal L}^0(\R)$, VaR with a \emph{confidence level} $\alpha \in (0,1)$ is given by $\operatorname{VaR}_\alpha(X)=\inf\{x \in \R: F_X(x)\ge \alpha\}$.
Moreover, for $X\in {\cal L}^1(\R)$, ES with a confidence level $\alpha \in (0,1)$ is $\operatorname{ES}_\alpha(X)=(1/(1-\alpha))\int_\alpha^1 \operatorname{VaR}_\beta(X){\rm d}\beta$, which coincides with
$\operatorname{ES}_\alpha(X)=\mathbb{E}[X|X\ge \operatorname{VaR}_\alpha(X)]$ if $X\in \mathcal L_\text{c}^1(\mathbb{R})$.
The confidence level $\alpha$ is typically chosen to be close to $1$, such as $0.975$ and $0.99$.
Let $\bX=(X_1,\dots,X_d)^\top\in{\cal L}^0(\R^d)$ be a random vector standing for the collection of losses to a portfolio or a financial institution.
Moreover, let $S=X_1+\cdots+X_d$ be the total loss.
Our sign convention is that positive values are losses and negative values are profits.

Under the prevalent Euler principle, the contribution of each component $X_j$, $j=1\dots,d$, to the risk of the total loss $\operatorname{ES}_\alpha(S)$ is determined by the gradient
$(\partial/\partial \lambda_j)\operatorname{ES}_\alpha(\utwi{\lambda}^\top \bX)|_{\boldsymbol{\lambda}=\bone_d}$ provided that the partial derivative exists, where $\utwi{\lambda}=(\lambda_1,\dots,\lambda_d)^\top\in\mathbb{R}^d$ and  $\bone_d=(1,\dots,1)^\top\in \mathbb{R}^d$.
Given certain smoothness conditions~\citep{tasche1999risk}, this derivative leads to:
\begin{align}\label{eq:esc}
\operatorname{ESC}_\alpha(X_j,S)=\mathbb{E}[X_j|S\ge \operatorname{VaR}_\alpha(S)],
\end{align}
which we call ESC for the $j$th risk.
Note that ESC itself is well-defined for $(X_j,S)\in{\cal L}^1(\R^2)$ without smoothness assumptions.
If $S\in \mathcal L_\text{c}^1(\mathbb{R})$, then the following \emph{full allocation property (FAP)} holds:
\begin{align}\label{eq:full:allocation}
\sum_{j=1}^d  \operatorname{ESC}_\alpha(X_j,S) =\operatorname{ES}_\alpha(S),
\end{align}
for which the total capital $\operatorname{ES}_\alpha(S)$ is allocated to $d$ components by the $d$-tuple of allocated capital:
\begin{align*}
    \operatorname{ESC}_\alpha(\bX,S)=(\operatorname{ESC}_\alpha(X_1,S),\dots,\operatorname{ESC}_\alpha(X_d,S)).
\end{align*}
For a generic $k$-dimensional risk functional  $\boldsymbol{\varrho}\in \{\operatorname{VaR}_\alpha,\operatorname{ES}_\alpha,\operatorname{ESC}_\alpha\}^k$, we write $\boldsymbol{\varrho}(F)$ for $\boldsymbol{\varrho}(\bX)$, $\bX\sim F$, by the law-invariance of $\boldsymbol{\varrho}$.

\subsection{Multi-objective elicitability and joint identifiability of ES contributions}\label{sec:elicitability:identifiability}

We introduce several properties required for backtesting ESCs.
Let ${\cal F}\subseteq {\cal F}^0(\R^d)$, $k,\,m\in \mathbb{N}$, and $A\subseteq \mathbb{R}^k$.
A function $\mathbf{S}:A\times \mathbb{R}^d \rightarrow \R^m$ is called \emph{$\mathcal F$-integrable} if, for every $F\in \mathcal F$ and $\bm{r}\in A$, every component of the function $\bx\mapsto \mathbf{S}(\bm{r},\bx)$ is integrable with respect to $F$.
A functional $\boldsymbol{\varrho}:{\cal F} \rightarrow A$ is called \emph{multi-objective elicitable} on ${\cal F}$ with respect to a total order $\preceq$ on $\R^m$ if there exists an $\mathcal F$-integrable function $\mathbf{S}:A\times \mathbb{R}^d \rightarrow \R^m$ such that, for every $F\in{\cal F}$, $\boldsymbol{\varrho}(F)$ is the unique minimizer of $\bm{r}\mapsto \mathbb{E}[\mathbf{S}(\bm{r},\bX)]$, $\bX\sim F$, over $A$.
The function $\mathbf{S}$ is called a \emph{(strictly $\mathcal F$-consistent multi-objective) scoring function} for $\boldsymbol{\varrho}$.
We simply call $\boldsymbol{\varrho}$ elicitable when a scoring function can be taken with $m=1$.

A functional $\boldsymbol{\varrho}:{\cal F} \rightarrow A$ is called \emph{identifiable} on ${\cal F}$ if there exists an $\mathcal F$-integrable function $\mathbf{V}:A\times \mathbb{R}^d \rightarrow \R^m$ such that, for every $F\in{\cal F}$, $\boldsymbol{\varrho}(F)$ is the unique solution to the equation $\mathbb{E}[\mathbf{V}(\bm{r},\bX)]=\bzero_m$, $\bX\sim F$, in terms of $\bm{r}$ on $A$.
The function $\mathbf{V}$ is called a \emph{(strict $\mathcal F$-) identification function} for $\boldsymbol{\varrho}$.

For $j\in \{1,\dots,d\}$, the $j$th ESC~\eqref{eq:esc} coincides with the \emph{marginal expected shortfall (MES)} of $(S,X_j)^\top$ considered in~\citet{fissler2024backtesting} provided that $(S,X_j)^\top\in \mathcal F_\text{c}^1(\mathbb{R}^2)$.
This observation immediately yields the following results on the multi-objective elicitability and joint identifiability of ESCs shown in Theorem~4.2 and Theorem~S.3.1 of ~\citet{fissler2024backtesting}, respectively.
In the following, the \emph{lexicographic order} $\leq_{\text{lex}}$ is adopted as a total order on $\R^2$.
For every $(a_1,b_1),(a_2,b_2)\in \mathbb{R}^2$, we write $(a_1,b_1)\leq_{\text{lex}}(a_2,b_2)$  if $a_1 < a_2$ or if $(a_1=a_2\text{ and }b_1\le b_2)$.
Moreover, we define:
\begin{align}\label{eq:sets:candidates}
\tilde{\mathcal F}_{\text{c}}^1(\mathbb{R}^{d})=
\left\{F \in \mathcal F^1(\mathbb{R}^{d}): F_{X_j,S}\in  {\cal F}_\text{c}^1(\mathbb{R}^2)\text{ for }\bX\sim F\text{ and }S=\sum_{j=1}^d X_j\right\}.
\end{align}

\begin{proposition}\label{prop:elicitability:esc}
Let $\alpha \in (0,1)$ and $\mathcal F\subseteq \tilde{\mathcal F}_{\text{c}}^1(\mathbb{R}^{d})$.
    \begin{enumerate}
        \item[(S1)] For every $j\in\{1,\dots,d\}$, the pair $F_{\bX}\mapsto (\operatorname{ESC}_\alpha(F_{X_j,S}),\operatorname{VaR}_\alpha(F_S))$ is multi-objective elicitable on $\mathcal F$ with respect to $(\mathbb{R}^2,\le_{\text{lex}})$. A strictly $\mathcal F$-consistent multi-objective scoring function $\mathbf{S}_j:\mathbb{R}^2\times \mathbb{R}^{d} \rightarrow (\mathbb{R}^2,\le_{\text{lex}})$ is given by:
        \begin{align*}
\mathbf{S}_j((m_j,v),\bx)&=(\S^\text{VaR}(v,s),\S_j^\text{ESC}((m_j,v),\bx))^\top,\\
&    \S^\text{VaR}(v,s)=
\{\id{\{s\le v\}}-\alpha\}\{h(v)-h(s)\},\\
& \S_j^\text{ESC}((m_j,v),\bx)=\id{\{s > v\}} \left\{
\phi_j'(m_j)(m_j-x_j) - \phi_j(m_j)+\phi_j(x_j)
\right\},
\end{align*}
where $s=\sum_{i=1}^d x_i$, $h:\R\rightarrow \R$ is a strictly increasing function, and $\phi_j:\R\rightarrow \R$ is a strictly convex differentiable function with derivative $\phi_j'$ such that $\mathbf{S}_j$ is $\mathcal F$-integrable.
        \item[(S2)] The $(d+1)$-tuple $F_{\bX}\mapsto (\operatorname{ESC}_\alpha(F_{\bX,S}),\operatorname{VaR}_\alpha(F_S))$ is multi-objective elicitable on $\mathcal F$ with respect to $(\mathbb{R}^2,\le_{\text{lex}})$.
        A strictly $\mathcal F$-consistent multi-objective scoring function $\mathbf{S}:\mathbb{R}^{d+1}\times \mathbb{R}^{d} \rightarrow (\mathbb{R}^2,\le_{\text{lex}})$ is given by:
\begin{align*}
\mathbf{S}((\bm{m},v),\bx)=\left(\S^\text{VaR}(v,s),\sum_{j=1}^d\S_j^\text{ESC}((m_j,v),\bx)\right)^\top,
\end{align*}
where $\S^\text{VaR}$ and $\S_j^\text{ESC}$, $j=1,\dots,d$, are as defined in (S1).
    \end{enumerate}
\end{proposition}

\begin{proposition}\label{prop:joint:identifiability:esc}
Let $\alpha \in (0,1)$ and $\mathcal F\subseteq \tilde{\mathcal F}_{\text{c}}^1(\mathbb{R}^{d})$.
    \begin{enumerate}
        \item[(V1)] For every $j\in\{1,\dots,d\}$, the pair $F_{\bX}\mapsto (\operatorname{ESC}_\alpha(F_{X_j,S}),\operatorname{VaR}_\alpha(F_S))$ is identifiable on $\mathcal F$ with a strict $\mathcal F$-identification function $\mathbf{V}_j:\mathbb{R}^2\times \mathbb{R}^d \rightarrow \mathbb{R}^2$ given by:
        \begin{align*}
\mathbf{V}_j((m_j,v),\bx)&=(\V^\text{VaR}(v,s),\V_j^\text{ESC}((m_j,v),\bx))^\top,\\
&    \V^\text{VaR}(v,s)=\alpha-\id{\{s\le v\}},\\
& \V_j^\text{ESC}((m_j,v),\bx)=\id{\{s > v\}}
(x_j-m_j).
\end{align*}
        \item[(V2)] The $(d+1)$-tuple $F_{\bX}\mapsto (\operatorname{ESC}_\alpha(F_{\bX,S}),\operatorname{VaR}_\alpha(F_S))$ is identifiable on $\mathcal F$ with a strict $\mathcal F$-identification function $\mathbf{V}:\mathbb{R}^{d+1}\times \mathbb{R}^{d} \rightarrow \mathbb{R}^{d+1}$ given by:
        \begin{align*}
\mathbf{V}((\bm{m},v),\bx)&=(\V^\text{VaR}(v,s),\V_1^\text{ESC}((m_1,v),\bx),\dots,\V_d^\text{ESC}((m_d,v),\bx))^\top,
\end{align*}
where $\V^\text{VaR}$ and $\V_j^\text{ESC}$, $j=1,\dots,d$, are as defined in (V1).
    \end{enumerate}
\end{proposition}

\begin{remark}
To simplify the notation and subsequent discussion, in Proposition~\ref{prop:elicitability:esc} we specialize the form of scoring function originally obtained in Theorem~4.2 of ~\citet{fissler2024backtesting}.
For instance, we take a specific auxiliary function in $\S^{\text{VaR}}$ so that $\S^{\text{VaR}}$ is independent of $x$ and equal to the scoring function of VaR presented in Equation~(5) of \citet{ehm2016quantiles}.
\end{remark}

As a concrete example of the scoring functions in Proposition~\ref{prop:elicitability:esc}, the choice $h(s)=s$ in $\S^\text{VaR}$ yields the well-known \emph{pinball loss} $\S^\text{VaR}(v,s)=(\id{\{s\le v\}}-\alpha)(v-s)$, and $\phi_j(x)=x^2$ in $\S_j^\text{ESC}$ leads to the \emph{squared loss} $\S_j^\text{ESC}((m_j,v),\bx)=\id\{s > v\}(x_j-m_j)^2$;~see~\citet{nolde2017elicitability}~and~\citet{fissler2024backtesting} for more examples.
Finally, we defer some technical discussion on the scoring and identification functions in Section~\ref{sec:order:sensitivity}.

\subsection{Backtesting dynamic ES contributions}\label{sec:backtesting:dynamic:ES:cont}

We introduce a setting for estimating dynamic ESCs.
Let $\{\bX_t\}_{t \in \N}$, $\bX_t=(X_{1,t},\dots,X_{d,t})^\top$, be a series of losses of interest, which is adapted to the filtration ${\cal G}=\{{\cal G}_t\}_{t\in \N}$.
For $S_t=X_{1,t}+\cdots +X_{d,t}$, let $VaR_t$ and $ES_t$ be VaR and ES of $S_{t}|{\cal G}_{t-1}$, respectively, with the prescribed confidence level $\alpha$.
For $j=1,\dots,d$, the $j$th ESC of $\bX_t$ given ${\cal G}_{t-1}$ is denoted by $ESC_{j,t}$.

We introduce a generic notation $\utwi{\varrho}_t$ for the pair $(ESC_{j,t},VaR_t)$ for a fixed $j\in \{1,\dots,d\}$, or the $(d+1)$-tuple $(ESC_{1,t},\dots,ESC_{d,t},VaR_t)$.
Let $\{\hat{\utwi{\varrho}}_t\}_{t\in \mathbb{N}}$ and $\{\hat{\utwi{\varrho}}_t^\ast\}_{t\in \mathbb{N}}$ be two series of ($\mathcal G$-predictable) forecasts of $\{\utwi{\varrho}_t\}_{t\in \mathbb{N}}$.
We assume that $\tilde{\mathcal F}_{\text{c}}^1(\mathbb{R}^{d+1})$ in~\eqref{eq:sets:candidates} contains all distributions of $(\bX_t,S_t)$ and $(\bX_t,S_t)|\mathcal G_{t-1}$ almost surely.
The multi-objective elicitability and identifiability presented in Section~\ref{sec:elicitability:identifiability} enable the conduct of \emph{comparative backtests} and \emph{traditional backtests (calibration tests)} based on the statistics
$\mathbf{\bar S}(\hat{\utwi{\varrho}}) = (1/|\mathbb{T}_{\text{out}}|)\sum_{t\in\mathbb{T}_{\text{out}}}\mathbf{S}(\hat{\utwi{\varrho}}_t,\bX_t)$ and 
$\mathbf{\bar V}(\hat{\utwi{\varrho}}) = (1/|\mathbb{T}_{\text{out}}|)\sum_{t\in\mathbb{T}_{\text{out}}}\mathbf{V}(\hat{\utwi{\varrho}}_t,\bX_t)$, respectively,
with $\mathbf{\bar S}(\hat{\utwi{\varrho}}^\ast)$ and $\mathbf{\bar V}(\hat{\utwi{\varrho}}^\ast)$ defined analogously, where $\mathbf{S}$ is a multi-objective scoring function for $\utwi{\varrho}$, $\mathbf{V}$ is an identification function for $\utwi{\varrho}$, and $\mathbb{T}_{\text{out}}$ is the \emph{out-of-sample period} $\mathbb{T}_{\text{out}}$ such that we forecast and backtest $\{\utwi{\varrho}_t\}_{t\in \mathbb{T}_{\text{out}}}$.
In practice, we divide a given finite sample period $\mathbb{T} = \{1,\dots,n+T\}$ into the \emph{in-sample period} $\mathbb{T}_{\text{in}} =\{1,\dots,n\}$ and the out-of-sample period $\mathbb{T}_{\text{out}} =\{n+1,\dots,n+T\}$.
For each $t\in \mathbb{T}_{\text{out}}$, we forecast $\utwi{\varrho}_t$ at $t-1$ based on the past $n$ observations of $\{\bX_{t-s}\}_{s \in \mathbb{T}_{\text{in}}}$.

Comparative backtests of ESCs concern whether some order between $\mathbf{\bar S}(\hat{\utwi{\varrho}})$ and $\mathbf{\bar S}(\hat{\utwi{\varrho}}^\ast)$ is statistically supported, and \emph{one-step} and \emph{two-step approaches} are proposed due to multi-objective elicitability with respect to the lexicographical order; see Section~5 of \citet{fissler2024backtesting} for details.
In traditional backtests of ESCs, we seek statistical evidence on the signs of $\mathbf{\bar V}(\hat{\utwi{\varrho}})$ and $\mathbf{\bar V}(\hat{\utwi{\varrho}}^\ast)$.
We refer the reader to Section 2.2 of \citet{nolde2017elicitability} for details.

\section{Robust forecast evaluation of ES contributions}\label{sec:robust}

To conduct a comparative backtest as presented in Section~\ref{sec:backtesting:dynamic:ES:cont}, we select a (multi-objective) scoring function from the class of functions presented in Proposition~\ref{prop:elicitability:esc}.
This is also required when we estimate models via score minimization, which we will consider in Section~\ref{sec:proposed:model}.
Such a dependence of the risk measurement procedure on the choice of scoring functions complicates the fair evaluation of forecasts.
\citet{patton2020comparing} also shows that forecast rankings can be sensitive to the choice of scoring function.

To overcome this issue,~\citet{ehm2016quantiles} propose a diagnostic tool for evaluating multiple forecasts called the \emph{Murphy diagram}, which is based on a mixture representation of a relevant class of scoring functions.
According to Theorem~1 of~\citet{ehm2016quantiles}, $\S^{\text{VaR}}$ in Proposition~\ref{prop:elicitability:esc} admits the mixture representation $\S^\text{VaR}(v,s)=\int_{\R}\S_\eta^\text{VaR}(v,s)\mathrm{d}H(\eta)$ for a non-negative measure $H$, where $\S_\eta^\text{VaR}:\mathbb{R}\times \mathbb{R} \rightarrow \mathbb{R}$, defined by:
\begin{align}\label{eq:var:elementary}
\S_\eta^\text{VaR}(v,s)=(\id{\{s <v\}}-\alpha)(\id{\{\eta <v\}}-\id{\{\eta <s\}}),\quad \eta\in\R,
\end{align}
is called an \emph{elementary scoring function} for VaR.
For competing forecasts $\{\widehat{VaR}_t^{(\ell)}\}_{t\in\mathbb{T}_{\text{out}}}$, $\ell=1,\dots,L$, of VaRs, a Murphy diagram displays the curves
 $\eta \mapsto (1/|\mathbb{T}_{\text{out}}|)\sum_{t\in\mathbb{T}_{\text{out}}}\S_\eta^{\text{VaR}}\left(\widehat{VaR}_t^{(\ell)},s_t\right)$,  $\ell=1,\dots,L$,
against $\eta$.
The diagram is negatively oriented in the sense that  forecasts with lower curves are evaluated to be more accurate.

The next proposition provides mixture representations of the scoring functions for ESCs presented in Proposition~\ref{prop:elicitability:esc}.

\begin{proposition}\label{prop:esc:mixture}
\begin{enumerate}
\item[(M1)]
Fix $j\in\{1,\dots,d\}$.
The scoring function $\S_j^\text{ESC}$ in Proposition~\ref{prop:elicitability:esc} (S1) admits the mixture representation:
\begin{align}\label{eq:esc:mixture}
\S_j^\text{ESC}((m_j,v),\bx)=\int_{\R}\S_{j,\eta}^\text{ESC}((m_j,v),\bx)\,\mathrm{d}H_j(\eta),
\end{align}
where $H_j$ is a non-negative measure satisfying
$\mathrm{d}H_j(\eta)=\mathrm{d}\phi_j'(\eta)$, $\eta\in\mathbb{R}$, and $\S_{j,\eta}^\text{ESC}:\mathbb{R}^2\times \mathbb{R}^d\rightarrow \mathbb{R}$ is defined by:
\begin{align}\label{eq:esc:elementary}
    \S_{j,\eta}^\text{ESC}((m_j,v),\bx)=\begin{cases}
\id{\{s>v\}}(x_j-\eta),&\text{ if }m_j\le\eta<x_j,\\
\id{\{s>v\}}(\eta-x_j),&\text{ if }x_j\le\eta<m_j,\\
0,&\text{ otherwise}.\\
\end{cases}
\end{align}
\item[(M2)] The scoring function $\S^\text{ESC}$ in Proposition~\ref{prop:elicitability:esc} (S2), with $\phi_j=:\phi$ for $j=1,\dots,d$, admits the mixture representation:
\begin{align}\label{eq:esc:mixture:all}
\S^\text{ESC}((\bm{m},v),\bx)=\int_{\R}\S_{\eta}^\text{ESC}((\bm{m},v),\bx)\,\mathrm{d}H(\eta),
\end{align}
where $H$ is a non-negative measure satisfying
$\mathrm{d}H(\eta)=\mathrm{d}\phi'(\eta)$, $\eta\in\mathbb{R}$, and $\S_{\eta}^\text{ESC}:\mathbb{R}^2\times \mathbb{R}^d\rightarrow \mathbb{R}$ is defined by:
\begin{align*}
    \S_{\eta}^\text{ESC}((\bm{m},v),\bx)=\sum_{j=1}^d \S_{j,\eta}^\text{ESC}((m_j,v),\bx).
    \end{align*}
\end{enumerate}
\end{proposition}

\begin{remark}\label{rem:interpretation}
Following Section~2.3 of~\citet{ehm2016quantiles}, the elementary scoring function~\eqref{eq:esc:elementary} can be interpreted as a degree of regret for the $j$th branch of a company who has a fixed capital $\eta$ to cover a future loss in distress, whose point forecast is $m_j$; see Section~\ref{sec:interpretation} for details.
\end{remark}

Since $\phi_j$ and $\phi$ in (M1) and (M2) above are strictly convex, the corresponding measures $H_j$ and $H$ assign positive mass to any finite interval on $\mathbb{R}$.
Note that $\mathrm{d}H(x)=2\mathrm{d}x$ for the squared loss, and $\S_{j,\eta}^{\text{ESC}}$ arises from  $\S_{j}^{\text{ESC}}$ in Proposition~\ref{prop:elicitability:esc} (S1) by taking $\phi_j(x)=(x-\eta)_{+}$ although this function is not strictly convex.

Based on the mixture representations~\eqref{eq:esc:mixture} and~\eqref{eq:esc:mixture:all}, a Murphy diagram of ESCs can be drawn analogously to the case of VaR.
For $l\in\{1,\dots,L\}$, let $\{\hat{\utwi{\varrho}}_t^{(\ell)}\}_{t\in \mathbb{T}_{\text{out}}}$, $\hat{\utwi{\varrho}}_t=(\widehat{ESC}_{1,t}^{(\ell)},\dots,\widehat{ESC}_{d,t}^{(\ell)},\widehat{VaR}_t^{(\ell)})$, be a series of predictions in the setting of Section~\ref{sec:backtesting:dynamic:ES:cont}.
The Murphy diagram of the $j$th ESC displays the curves.
\begin{align*}
\eta \mapsto \bar \S_{j,\eta}^{\text{ESC}}(l):=\frac{1}{|\mathbb{T}_{\text{out}}|}\sum_{t\in\mathbb{T}_{\text{out}}}\S_{j,\eta}^{\text{ESC}}((\widehat{ESC}_{j,t}^{(\ell)},\widehat{VaR}_t^{(\ell)}),\bx_t), \quad \ell=1,\dots,L.
\end{align*}

The Murphy diagram of the $d$-tuple of ESCs analogously exhibits
$\eta \mapsto \bar \S_{\eta}^{\text{ESC}}(l):=\sum_{j=1}^d \bar \S_{j,\eta}^{\text{ESC}}(l)$ for $\ell=1,\dots,L$.
To obtain these curves and their differences on the whole real line, it suffices to evaluate them on a finite set because $\bar \S_{j,\eta}^{\text{ESC}}(l)$ and $\bar \S_{\eta}^{\text{ESC}}(l)$ are piecewise linear functions of $\eta$ with all kinks and jump points contained in $\mathcal D_j^{(l)}=\{x_{j,t},\widehat{ESC}_{j,t}^{(\ell)},t\in \mathbb{T}_{\text{out}}\}$ and $\mathcal D^{(l)}=\bigcup_{j=1}^d \mathcal D_j^{(l)}$, and $\bar \S_{j,\eta}^{\text{ESC}}(l)$ and $\bar \S_{\eta}^{\text{ESC}}(l)$ vanish outside of $[\min(\mathcal D_j^{(l)}),\max(\mathcal D_j^{(l)})]$ and $[\min(\mathcal D^{(l)}),\max(\mathcal D^{(l)})]$, respectively.

Note that Murphy diagrams of ESCs depend on the corresponding forecasts of total VaR.
Consequently, we can conduct a completely fair evaluation among multiple forecasts of ESCs when they share the same forecasts of total VaR.
In such a case, more advanced analyses, such as forecast dominance tests proposed by~\citet{ziegel2020robust}, are available.

\section{Proposed models for estimating dynamic ES contributions}\label{sec:proposed:model}

In previous sections we have mentioned that clear and rigid comparison among forecasting models of ESCs is feasible among those with common forecasts of total VaRs.
In this section we exploit this feature and propose new models based on compositinal regression.

According to the evaluation criterion based on multi-objective elicitability of ESCs, forecast accuracy of total VaR is of utmost importance since estimated ESCs are compared only for models estimating total VaRs with equal accuracy.
As considered in \citet{dimitriadis2023dynamic}, this valuation principle naturally encourages a two-stage approach, where we first model the dynamics of $\{VaR_t\}$, possibly induced from the dynamics of $\{S_t\}$, and then consider the dynamics of ESCs, which may also be induced from the dynamics of $\{(X_j,S_t)\}$ for $j=1,\dots,d$.
We summarize benefits and limitations of such a \emph{top-down approach} in Section~\ref{sec:top:bottom}.
In this approach the dynamics of $\{ES_t\}$ can also be estimated together with $\{VaR_t\}$ in the first stage due to the natural order in the risk measurement procedure and, more importantly, joint elicitability of $(\operatorname{VaR}_\alpha,\operatorname{ES}_\alpha)$, where ES$_\alpha$ can be elicited only in combination with VaR$_\alpha$~\citep{fissler2016higher}.
An alternative approach is to specify the dynamics of $\{\bX_t\}$, which we call a \emph{bottom-up approach}.
This approach may be discouraged in terms of the two-stage forecast evaluation of ESCs since total VaR is modeled only indirectly;~see Section~\ref{sec:top:bottom} for more details of this approach.

Despite the appeals of the top-down approach in forecasting ESCs, special attention is required in this approach so that the empirical counterpart of FAP $\sum_{j=1}^d \widehat{ESC}_{j,t}=\widehat{ES}_{t}$ is satisfied.
To handle this constraint, we propose new models for estimating dynamic ESCs.
In what follows, we describe our proposed procedure for the case of one-step ahead forecast, where we use the losses $\{\bX_t\}_{t\in \mathbb{T}_{\text{in}}}$, $\mathbb{T}_{\text{in}}=\{1,\dots,n\}$, to forecast ESCs at time $n+1$.

The proposed model consists of two stages, where the dynamics of total ES is first estimated in combination with that of total VaR, and then ESCs are estimated
from the proportion of the total ES to each component of risk.
In the first stage, we estimate the dynamics of $(ES_t,VaR_t)$ for $t \in \mathbb{T}_{\text{in}}$ and then forecast $(VaR_{n+1},ES_{n+1})$ based on this estimated dynamics.
For this purpose, various models are proposed in the literature; see, for example,~\citet{patton2019dynamic},~\citet{taylor2019forecasting} and~\citet{taylor2022forecasting}.
Joint elicitability of the pair of risk measures $(\operatorname{VaR}_\alpha,\operatorname{ES}_\alpha)$ enables us to estimate the joint dynamics of this pair by score minimization; see, for example, \citet{taylor2019forecasting}.
In this paper we do not specify any specific model in this stage and instead select the best model based on the currently available information and target models for comparison.
Denote by $(\widehat{VaR}_t,\widehat{ES}_t)_{t\in\mathbb{T}}$, $\mathbb{T}=\{1,\dots,n+1\}$, the dynamics forecasted in this stage.

In the second stage, we allocate $\widehat{ES}_t$ to each component of risk with the weight vector $\bm{w}_t = (w_{1,t},\dots,w_{d,t})^\top\in\mathbb{R}^d$, $\sum_{j=1}^d w_{j,t}=1$, such that $ESC_{j,t}=w_{j,t}\times  ES_t$.
We assume that $ESC_{j,t}>0$ for all $j=1,\dots,d$ and $t\in\mathbb{T}$.
This assumption implies that $\bm{w}_t\in(0,1)^d$ by the \emph{diversifying property} $ESC_{j,t}\le ES_t$~\citep{kalkbrener2005axiomatic}. 
We further discuss this assumption and its relaxation in Remark~\ref{rem:negative:esc}.
Consequently, the weight vector $\bm{w}_t$ lies in $\mathcal S_{d}=\left\{\bm{w} \in (0,1)^d: \sum_{j=1}^d w_j= 1\right\}$.
Data on $\mathcal S_{d}$ are called \emph{compositional data}, and statistical modeling of such data has been extensively studied in the literature; see~\citet{aitchison1982statistical,aitchison2005compositional,pawlowsky2011compositional}; and references therein.
A common approach for modeling compositional data is first transforming the data to an unconstrained space to eliminate the sum constraint and then constructing a model on the unconstrained space.

One of the most prominent examples of such a transform is the \emph{isometric log-ratio transformation}~\citep[\emph{ilr,}][]{egozcue2003isometric} defined by:
\begin{align*}
    \operatorname{ilr}(\bm{w}) = V^\top \ln(\bm{w}) \in \mathbb{R}^{d-1},\quad \bm{w}\in \mathcal S_d,
\end{align*}
where $V\in \mathbb{R}^{d\times (d-1)}$ is a given matrix such that $V^\top V = I_{d-1}$, $V V^\top = I_d -\frac{1}{d}\bone_d\bone_d^\top$ and $V^\top \bone_d = \bzero_{d-1}$.
Note that the vector of equal weights $(1/d)\,\bone_d \in \mathcal S_d$ corresponds to the origin $\operatorname{ilr}((1/d)\,\bone_d)=\bzero_{d-1}\in\mathbb{R}^{d-1}$.
The matrix $V$ is called the \emph{contrast matrix}.
The map $\operatorname{ilr}$ is one-to-one, and its inverse map is given by:
\begin{align}\label{eq:closing}
    \operatorname{ilr}^{-1}(\bm{z})=\operatorname{C}(\exp(V \bm{z}))\in \mathcal S_d,\quad \bm{z}\in\mathbb{R}^{d-1},
\end{align}
where $\operatorname{C}:\mathbb{R}^d\rightarrow \mathcal S_d$ is the \emph{closing operator} defined by: 
\begin{align*}
\operatorname{C}(\bx)=\frac{\bx}{\bone_d^\top \bx},\quad \bx\in \mathbb{R}^d.
\end{align*}
Note that the map $\operatorname{ilr}$ depends on the choice of $V$, and we fix this matrix so that the resulting transform yields:
\begin{align*}
\operatorname{ilr}_k(\bm{w})=\sqrt{\frac{d-k}{d-k+1}}\ln\left(\frac{w_k}{\left(\prod_{l=k+1}^d w_l\right)^{1/(d-k)}}\right),
\quad k=1,\dots,d-1.
\end{align*}

With this choice of $V$, $\operatorname{ilr}_k(\bm{w})$, $k\in \{1,\dots,d-1\}$, can be interpreted as a normalized balance between the $k$th weight and the group of the weights over the assets $k+1,\dots,d$.
Note that $0$ and $1$ cannot be included in the compositional data for the above transformations to be well-defined.
This well-known limitation of ilr should not cause a significant problem in our analysis since $0$ or $1$ in the allocated weights corresponds to the practically irrelevant case when no capital or full capital is allocated to one asset, respectively.

We consider the following generic model on the dynamics of the allocation weights $\{\bm{w}_t\}_{t \in \mathbb{T}}$:
\begin{align}\label{eq:w:model}
\bm{w}_{t+1}&=\upsilon_{\utwi{\theta}}(\bm{w}_{s},\bX_s, s\le t)
:=
\operatorname{ilr}^{-1}\left(\utwi{\tau}+\utwi{\Phi}
\operatorname{ilr}(\bm{w}_{t})+
\utwi{\Psi} h(\bm{X}_s, s\le t)
\right),\quad t\in \mathbb{T},
\end{align}
where $h$ is a $\mathbb{R}^q$-valued function and $\utwi{\theta}=(\utwi{\tau},\utwi{\Phi},\utwi{\Psi})$ is a set of parameters on a parameter space $\Theta$ such that $\utwi{\tau}\in \mathbb{R}^{d-1}$, $\utwi{\Phi}\in \mathbb{R}^{(d-1)\times (d-1)}$, and $\utwi{\Psi}\in \mathbb{R}^{(d-1)\times q}$.
For brevity, we only regress $\bm{w}_{t+1}$ by $\bm{w}_t$ although one can choose more lagged variables as regressors.
The initial weight $\bm{w}_1$ can also be regarded as a parameter to be estimated or as an externally given constant.
We then estimate parameters in~\eqref{eq:w:model} by minimizing the average score of the $d$-tuple of ESCs in Proposition~\ref{prop:elicitability:esc} (S2) with $\phi_1=\cdots=\phi_d=:\phi$, $\phi(x)=x^2$:
\begin{align}\label{eq:theta:score}
\hat{\utwi{\theta}}=\operatorname{argmin}_{\utwi{\theta}\in \Theta}\sum_{t \in \mathbb{T}_{\text{in}}}\sum_{j=1}^d\left\{x_{j,t}-\upsilon_{\utwi{\theta}}(\bm{w}_s,\bx_s, s\le t)_j\widehat{ES}_t\right\}^2 \id{\left\{s_t > \widehat{VaR}_t\right\}},
\end{align}
where $\bx_t=(x_{1,t},\dots,x_{d,t})^\top$ and $s_t$ are the realizations of $\bX_t$ and $S_t$, respectively, and $\widehat{VaR}_t$ and $\widehat{ES}_t$ are the estimates given in the first stage.
Once the estimated parameter $\hat{\utwi{\theta}}$ is obtained, we forecast the ESCs at time $n+1$ by ${\widehat{ESC}}_{j,n+1}=\hat{w}_{j,n+1}\widehat{ES}_{n+1}$, $j=1,\dots,d$, where $\hat{\bm{w}}_{1}=\bm{w}_{1}$ and $\hat{\bm{w}}_{t+1}=\upsilon_{\hat{\utwi{\theta}}}(\hat{\bm{w}}_{s},\bX_s, s\le t)$, $t=1,\dots,n$.
Note that the estimated weight $\hat{\bm{w}}_{n+1}$ belongs to $\mathcal S_d$, which ensures the empirical counterpart of FAP: $\sum_{j=1}^d \widehat{ESC}_{j,t}=\widehat{ES}_{t}$.

\begin{remark}\label{rem:negative:esc}
If we assume that $ESC_{j,t}<ES_t$ instead of $ESC_{j,t}>0$, for all $j=1,\dots,d$ and $t\in\mathbb{T}$, then it holds that $-(d-2)\,ES_t < ESC_{j,t}$ by FAP.
Therefore, by defining the vector of the \emph{normalized weights} :
\begin{align*}
\tilde{\bm{w}}_{t}=(\tilde w_{1,t},\dots,\tilde w_{d,t})^\top\quad\text{where}\quad
\tilde w_{j,t}=\frac{ESC_{j,t}+(d-2)\,ES_t}{(d-1)^2\, ES_t},\quad j=1,\dots,d,
\end{align*}
we have that $\tilde{\bm{w}}_{t}\in \left(0,1/(d-1)\right)^d$ and that $\sum_{j=1}^d \tilde w_{j,t}=1$.
Since $\tilde{\bm{w}}_{t}$ belongs to a subset of $\mathcal S_d$, one could naively fit the model~\eqref{eq:w:model} with the range constraint $\tilde{\bm{w}}_{t}\in \left(0,1/(d-1)\right)^d$ possibly handled through penalization in the loss function in~\eqref{eq:theta:score}.
We do not explore this direction further in this study. 
Indeed, for the portfolio considered in Section~\ref{sec:empirical}, we observe that the forecasted ES and ESCs in the whole period are all positive.
Hence we adopt the assumption that $ESC_{j,t}>0$ for all $j=1,\dots,d$ and $t\in\mathbb{T}$.
\end{remark}
\section{Empirical study}\label{sec:empirical}

Various models in the literature are available to forecast ESCs.
In this section we compare them with our proposed model in Section~\ref{sec:proposed:model} in the traditional and comparative backtests presented in Section~\ref{sec:backtesting:dynamic:ES:cont} and based on the Murphy diagrams introduced in~Section~\ref{sec:robust}.
We first describe the setting and models in comparison in Section~\ref{sec:setting}.
We then show results of the forecast evaluation in Section~\ref{sec:evaluation}.
Finally, we discuss forecast accuracy of the compared models in Section~\ref{sec:discussion}.

\subsection{Setting and model description}\label{sec:setting}

We consider a portfolio consisting of the following $d=3$ stock prices with equal investment weights (which should not be confused with the ratios of ESCs to the total ES modeled in Section~\ref{sec:proposed:model}):
Amazon (AMZN),
Alphabet Class A (GOOGL), and Telsa (TSLA).
From 2010-06-30 to 2023-05-30, each series consists of $n+T=3249$ negative daily log returns multiplied by $100$.
We conduct a rolling window analysis with window size $n = 2249$ and forecast day-ahead ESCs for the last $T=1000$ observations.
Namely, for $t=n+1,\dots,n+T$, we forecast total VaR, total ES, and ESCs of $\bX_t|\mathcal G_{t-1}$ based on the past $n$ observations $\{\bX_{s}\}_{s\in \{t-n,\dots,t-1\}}$.
Following the \emph{Fundamental Review of the Trading Book (FRTB)}, we focus on the confidence levels $\alpha =0.975$~\citep{bcbs2013consultative}.

In this analysis we compare six models including the proposed one in Section~\ref{sec:proposed:model}.
We briefly introduce these models here and defer a detailed description to Section~\ref{sec:detailed:description}.
As a simple benckmark, our first model is the \emph{historical simulation (HS)} model, where at each time $t$ we estimate the risk functional $\utwi{\varrho}_t=(ESC_{1,t},\dots,ESC_{d,t},VaR_t,ES_t)$ nonparametrically based on the past $n$ observations.
Second, we consider what we call the \emph{bottom-up GARCH (GARCH.BU)} model, where estimates of the risk functionals are induced from a copula-GARCH model~\citep{jondeau2006copula,huang2009estimating} among $\{\bX_t\}$.
We also consider the \emph{top-down GARCH (GARCH.TD}) model, where we fit $d$ number of bivariate copula-GARCH models on $\{(X_{j,t},S_t)\}$, $j=1,\dots,d$.
Since the above three models do not take the hysteresis effect and time-varying correlations into account, our fourth model is the bivariate \emph{hysteretic autoregressive model with GARCH error and dynamic conditional correlations }~\citep[\emph{HAR.GARCH},~][]{chen2019quantile}, fitted to $\{(X_{j,t},S_t)\}$, $j=1,\dots,d$, in the top-down approach.
Following~\eqref{eq:full:allocation}, total ES is estimated as the sum of estimated ESCs in this model.

We next include two compositional regression models in our comparison.
These models share the estimates $\{\widehat{VaR}_t,\widehat{ES}_t\}$, which are obtained by an AR-GARCH model on $\{S_t\}$ used in the above GARCH.TD model for the purpose of comparing these models.
Our fifth model is termed the \emph{compositional regression model with least square estimation (CR.LSE)}.
In this model we follow~\citet{boonen2019forecasting} and first estimate a series of ESCs under the assumption that each $\bX_t|\mathcal G_{t-1}$ follows an elliptical distribution.
We then transform this series by~\eqref{eq:closing} to obtain the allocation weights $\{\hat{\bm{w}}_t\}\subset \mathcal S_d$.
By regarding this set of allocation weights as compositional data, we fit the compositional regression model~\eqref{eq:w:model} by the standard least square estimation.
Finally, our sixth model is the \emph{compositional regression model based on score optimization (CR.OPT)}, which is the proposed model described in Section~\ref{sec:proposed:model}.
To preserve interpretability of the model, we choose $\operatorname{ilr}$-transformed variables for $h(\bm{X}_t)$.
More specifically, we choose 
$h(\bm{X}_s, s\le t)=(h^{+}(\bm{X}_s, s\le t),h^{-}(\bm{X}_s, s\le t))^\top$ with 
\begin{align*}
    h^{+}(\bm{X}_s, s\le t)&=\operatorname{ilr}\circ\operatorname{C}\left(
\max\left(\frac{1}{t_0}\sum_{s=0}^{t_0-1}\bX_{t-s} S_{t-s}, \utwi{\epsilon}\right)\right),\\
h^{-}(\bm{X}_s, s\le t)&=\operatorname{ilr}\circ\operatorname{C}\left(-\min\left(
\frac{1}{t_0}\sum_{s=0}^{t_0-1}\bX_{t-s} S_{t-s},-\utwi{\epsilon}\right)\right),
\end{align*}
where $t_0=7$,
$\utwi{\epsilon}=(0.01,\dots,0.01)\in \mathbb{R}^d$ is a $d$-vector of small numbers, and $\max$ and $\min$ are applied component-wise.
The moving average term $(1/t_0)\,\sum_{s=0}^{t_0-1}\bX_{t-s} S_{t-s}$ is motivated by the fact that the allocation weight $\bm{w}_t$ equals $\operatorname{C}(\mathbb{E}[\bX_tS_t|{\cal G}_{t-1}])$ when $\bX_t|{\cal G}_{t-1}$ follows an elliptical distribution with zero vector of location parameters.
Therefore, we regard the corresponding regression coefficients of $h(\bm{X}_t)$ as the positive and negative effects of the local covariance between $\bX$ and $S$, calculated over the past $t_0$ days, to the balance among allocation weights.
In particular, for the reduced model:
\begin{align*}
\operatorname{ilr}(\bm{w}_{t+1})=
\operatorname{ilr}(\tilde{\utwi{\tau}})+
\phi \operatorname{ilr}(\bm{w}_t)+
\psi^{+}\operatorname{ilr}(h^{+}(\bm{X}_s, s\le t))+
\psi^{-}\operatorname{ilr}(h^{-}(\bm{X}_s, s\le t)),
\end{align*}
where $\tilde{\utwi{\tau}}=\operatorname{ilr}^{-1}(\utwi{\tau})$ and $\phi,\psi^{+},\psi^{-}\in \mathbb{R}$, the vector of weights $\bm{w}_{t+1}$ is given by the \emph{power-perturbation combination} of $\bm{w}_t$, $h^{+}(\bm{X}_s, s\le t)$ and $h^{-}(\bm{X}_s, s\le t)$, the counterpart of linear combination in a simplicial geometry; see, for example,~\citet{aitchison2005compositional}.
We extend this model by allowing mutual effects between $\operatorname{ilr}_1$ and $\operatorname{ilr}_2$.
We introduce $\utwi{\epsilon}$ to ensure that all the components of $\max(\bX_t S_t, \utwi{\epsilon})$ and $\min(
\bX_t S_t,-\utwi{\epsilon})$ are non-zero.
More sophisticated treatment with zeros can be possible; see~Chapter~4 of~\citet{pawlowsky2011compositional} and references therein.
For the initial allocation weight $\bm{w}_1$, we use the compositional data $\hat{\bm{w}}_1$ generated in CR.LSE.


In Figures~\ref{fig:ts:var:es} and~\ref{fig:ts:esc} we display the estimated dynamics of total VaR, total ES, and ESCs on the out-of-sample period $\mathbb{T}_{\text{out}}$ together with the realized losses.
Note that the forecasted VaR, ES and ESCs in the whole period are all greater than $0$.
To summarize our observations on the models, we find that HS only captures the average trends, HAR.GARCH tends to estimate ESCs lower than others, and CR.OPT typically  leads to more fluctuated estimates.
Compared with these three models, the other three models seem to produce estimates relatively similar with each other.


\begin{figure}[p]
\centering
\includegraphics[width=1.0\textwidth]{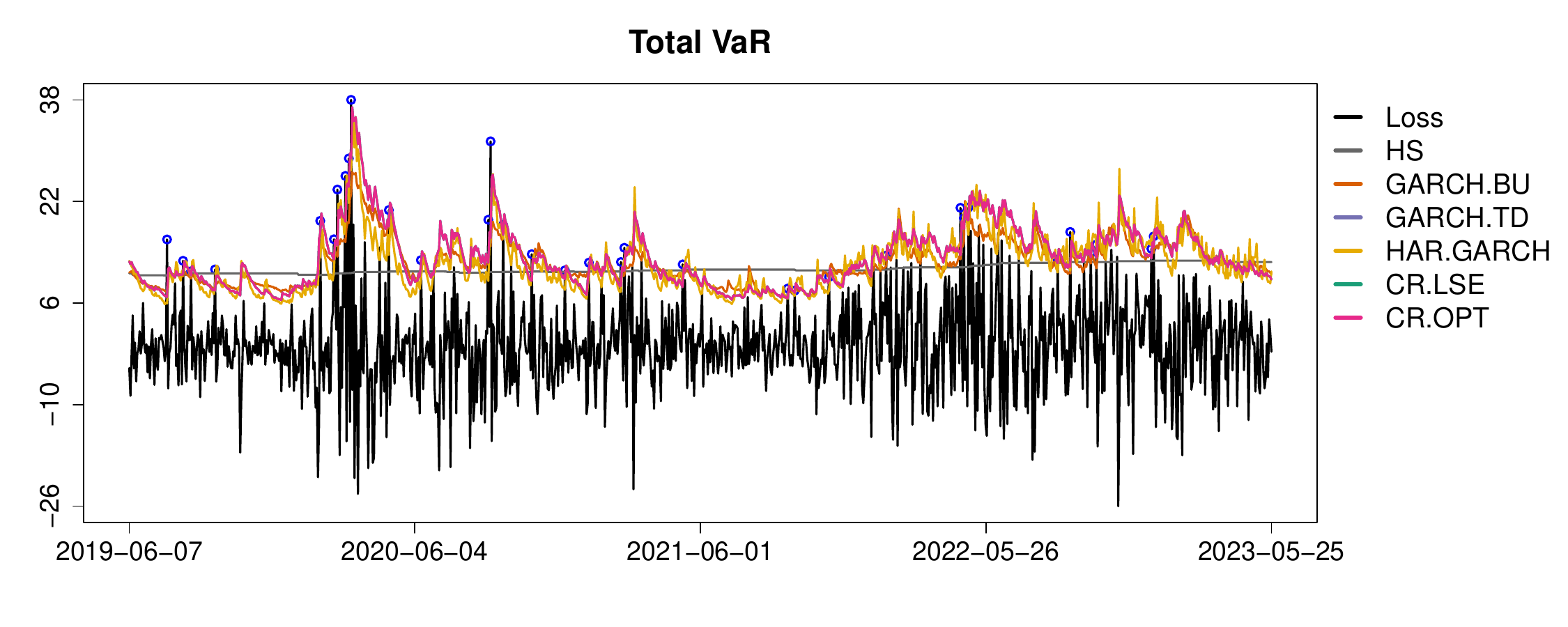}
\includegraphics[width=1.0\textwidth]{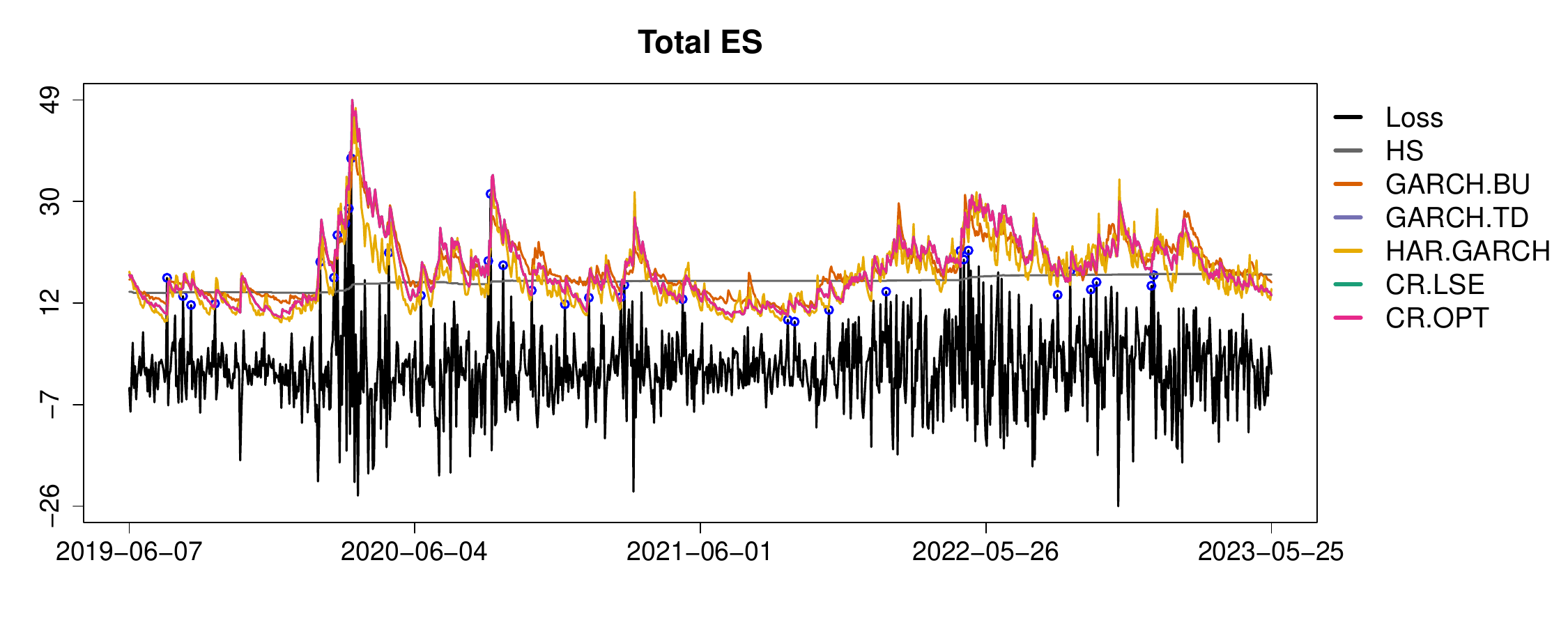}
\caption{Time series plots of total VaR and total ES with confidence level $0.975$ estimated by the six models.
The black lines represent the time series of the total loss of the portfolio, and the losses exceeding the estimated total VaRs used in CR.OPT are marked in blue.}
\label{fig:ts:var:es}
\end{figure}

\begin{figure}[p]
\centering
\includegraphics[width=1.0\textwidth]{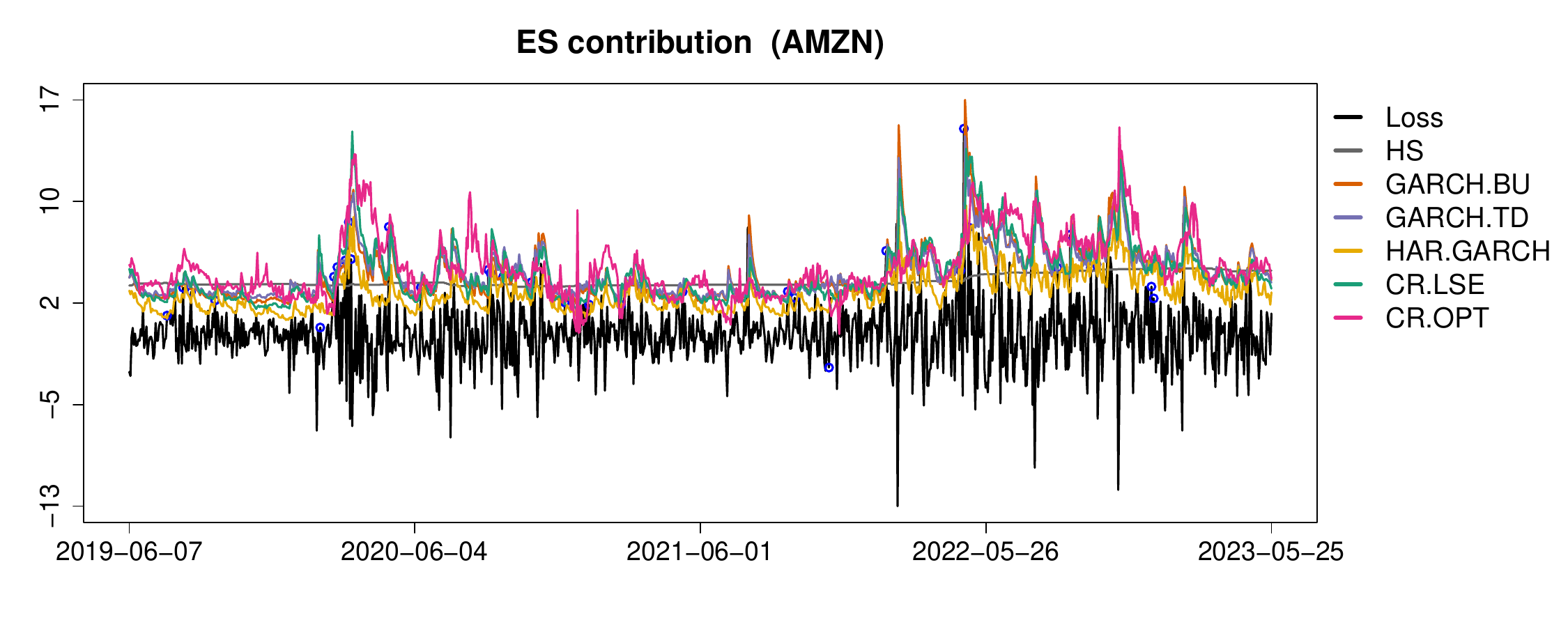}
\includegraphics[width=1.0\textwidth]{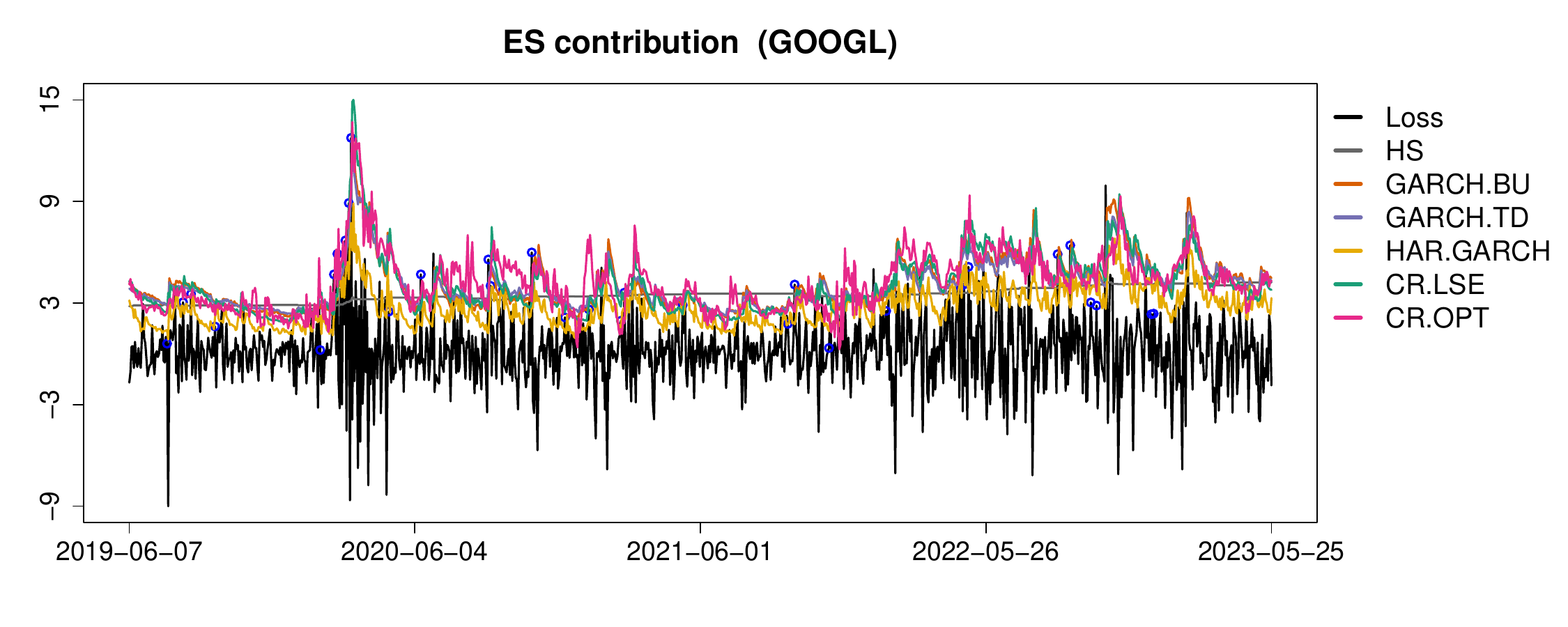}
\includegraphics[width=1.0\textwidth]{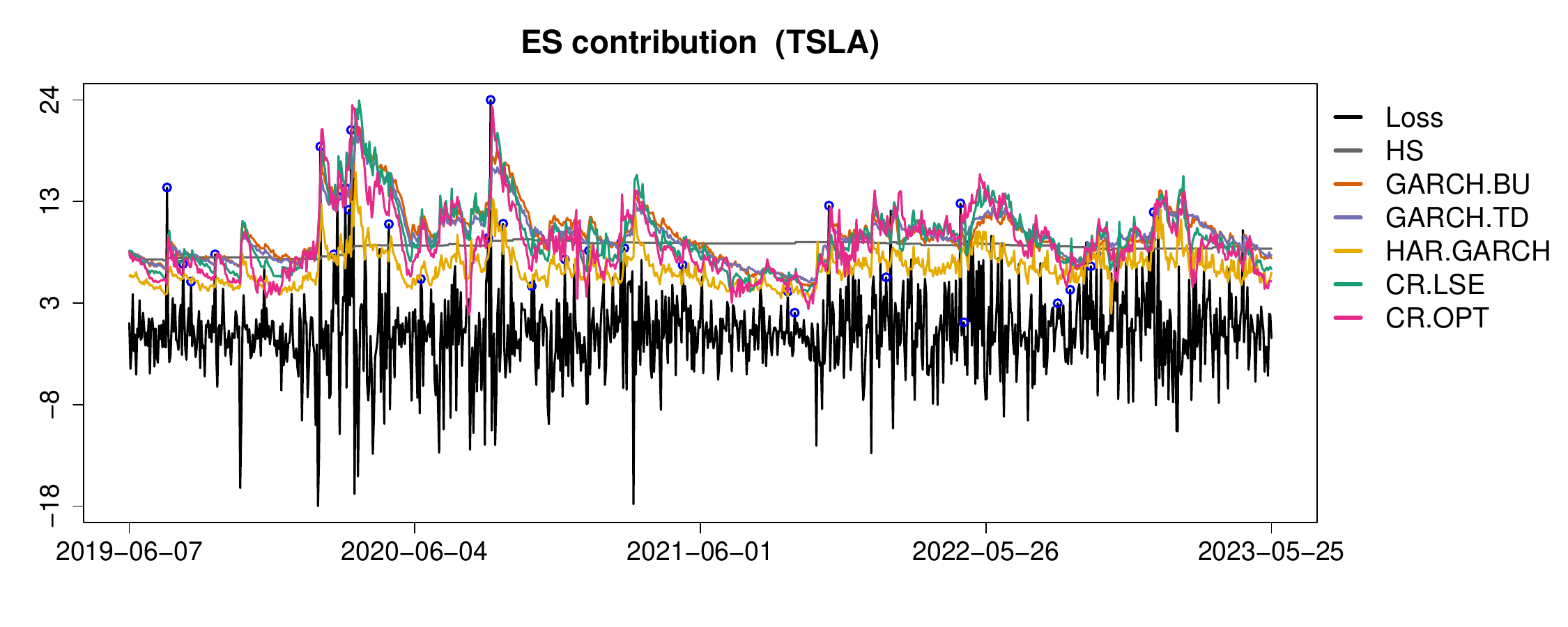}
\caption{Time series plots of ESCs with confidence level $0.975$ estimated by the six models.
The black lines represent the time series of the corresponding marginal losses, and the losses exceeding the estimated total VaRs used in CR.OPT are marked in blue.}
\label{fig:ts:esc}
\end{figure}

In Table~\ref{table:parameters} we provide summary statistics of the estimated parameters in CR.LSE and CR.OPT over the backtesting period. 
According to the table, the parameters estimated in CR.OPT are overall more volatile over time compared with those in CR.LSE.
For both of CR.LSE and CR.OPT, the diagonal elements of $\utwi{\Phi}$ tend to take large values than the off-diagonal elements.
In addition, the effect of $h(\bm{X}_s, s\le t)$ parametrized by $\utwi{\Psi}$ is estimated to be constantly small in CR.LSE whereas it is sometimes significant in CR.OPT.

\begin{table}[p]
\centering
\caption{
Summary statistics of the parameters of the model $\operatorname{ilr}(\bm{w}_{t+1})=
\utwi{\tau}+
\utwi{\Phi} \operatorname{ilr}(\bm{w}_t)+
\utwi{\Psi}^{+}\operatorname{ilr}(h^{+}(\bm{X}_s, s\le t))+
\utwi{\Psi}^{-}\operatorname{ilr}(h^{-}(\bm{X}_s, s\le t))$ over the backtesting period, where $\utwi{\tau}=(\tau_1,\tau_2)^\top$, $\utwi{\Phi}=(\phi_{ij})_{2\times 2}$, $\utwi{\Psi}^{+}=(\psi^{+}_{ij})_{2\times 2}$ and $\utwi{\Psi}^{-}=(\psi^{-}_{ij})_{2\times 2}$.}
\label{table:parameters}
\scalebox{0.9}{
\begin{tabular}{rrrrrrrrrrrr}
  \hline
& \multicolumn{5}{c}{CR.LSE} & & \multicolumn{5}{c}{CR.OPT}\\
  \hline
 & Mean & Median & SD  & 1st Qu. & 3rd Qu.  & & Mean & Median & SD  & 1st Qu. & 3rd Qu. \\ 
  \hline
  $\tau_1$ & -0.003 & -0.003 & 0.003 & -0.005 & -0.001 & &0.045 & 0.037 & 0.045 & 0.018 & 0.063 \\ 
  $\tau_2$ & -0.043 & -0.044 & 0.009 & -0.050 & -0.035 && -0.041 & -0.039 & 0.027 & -0.056 & -0.025 \\ 
  $\phi_{11}$ & 0.793 & 0.790 & 0.021 & 0.779 & 0.801 & &0.839 & 0.876 & 0.119 & 0.800 & 0.912 \\ 
  $\phi_{21}$ & 0.027 & 0.031 & 0.010 & 0.017 & 0.035 & &0.084 & 0.052 & 0.135 & -0.005 & 0.146 \\ 
  $\phi_{12}$ & -0.010 & -0.007 & 0.009 & -0.010 & -0.004& & -0.071 & -0.073 & 0.094 & -0.130 & -0.033 \\ 
 $\phi_{22}$ & 0.901 & 0.898 & 0.018 & 0.889 & 0.903 & &0.877 & 0.891 & 0.072 & 0.839 & 0.928 \\ 
  $\psi^{+}_{11}$ & 0.053 & 0.053 & 0.004 & 0.049 & 0.057 && -0.051 & -0.036 & 0.070 & -0.073 & -0.015 \\ 
  $\psi^{+}_{21}$ & 0.011 & 0.011 & 0.002 & 0.009 & 0.013 && 0.058 & 0.060 & 0.039 & 0.038 & 0.080 \\ 
  $\psi^{+}_{12}$ & 0.006 & 0.006 & 0.002 & 0.004 & 0.007 && -0.014 & -0.012 & 0.036 & -0.039 & 0.010 \\ 
  $\psi^{+}_{22}$ & 0.035 & 0.038 & 0.006 & 0.036 & 0.039 && 0.082 & 0.083 & 0.036 & 0.053 & 0.108 \\ 
  $\psi^{-}_{11}$ & 0.053 & 0.052 & 0.005 & 0.049 & 0.056 && -0.002 & 0.019 & 0.135 & -0.056 & 0.082 \\ 
 $\psi^{-}_{21}$ & 0.005 & 0.005 & 0.002 & 0.003 & 0.007 && 0.056 & 0.051 & 0.104 & -0.006 & 0.107 \\ 
  $\psi^{-}_{12}$ & 0.006 & 0.007 & 0.003 & 0.003 & 0.009 && 0.051 & 0.020 & 0.167 & -0.041 & 0.116 \\ 
  $\psi^{-}_{22}$ & 0.043 & 0.046 & 0.008 & 0.043 & 0.048  && 0.111 & 0.101 & 0.085 & 0.057 & 0.157 \\ 
   \hline
\end{tabular}}
\end{table}

\subsection{Model evaluation}\label{sec:evaluation}

We conduct the comparative and traditional backtests in the two-step approach, where we handle hypotheses on VaR (and ES) and those on ESCs separately; see Section~5.4 in \citep{fissler2024backtesting}.
In Section~\ref{sec:one:step:approach} we also conduct those in the alternative one-step approach~\citep{fissler2024backtesting}.
For the comparative backtests, we follow Section~2.3 of \citet{nolde2017elicitability} and conduct the \emph{Diebold-Mariano (DM)}-tests~\citep{diebold1995comparing} for the series of score differences.
For an \emph{HAC estimator}~\citep{andrews1991heteroskedasticity} in the test statistic, we choose the \emph{Bartlett kernel}~\citep{newey1987simple} with the automatic bandwidth estimator based on AR(1) approximation; see Section~6 of~\citet{andrews1991heteroskedasticity}.
In the context of~\citet{nolde2017elicitability}, we choose CR.OPT for the benchmark model and compare other models as internal models.
Consequently, in the three-zone approach of \citet{fissler2016expected}, the red region indicates that the superiority of CR.OPT over a compared model is statistically supported, the green region shows that CR.OPT is inferior to the alternative model, and the yellow region means that further investigation is required since there is no statistical evidence on the order between the two models.
We fix the significance level to be $0.05$.
We repeat this analysis for total VaR, total ES, tuple of ESCs, and $j$th ESC for $j=1,\dots,d$, although, in terms of multi-objective elicitability, the results on ESCs are not meaningful if forecast accuracy is ranked in total VaRs.
For the scoring functions, we choose the pinball loss for VaR, \emph{AL log score}~\citep{taylor2019forecasting} for ES, and squared loss for tuple of ESCs, as well as each of them.
Tables~\ref{table:test:var:es:esc} and~\ref{table:test:escs} report the average scores and p-values of the one-sided and two-sided DM-tests, based on which we obtain the regions in the three-zone approach.

\begin{table}[p]
\caption{Results of the DM-tests to compare the forecast accuracy of total VaR, total ES, and tuple of ESCs with CR.OPT as the benchmark model.
}
\label{table:test:var:es:esc}
\centering
\scalebox{0.9}{
\begin{tabular}{lrrrrrr}
  \hline
 & Average score$^a$ & Rank$^b$ &  \multicolumn{3}{c}{p-value$^c$} & Region$^d$
 \\  \hline
H0 & &  & $=$ CR.OPT & $\le$ CR.OPT & $\ge$ CR.OPT  & \\
  \hline
  \multicolumn{7}{l}{(1) Total VaR}\\
HS & 56.038 & 6 & {\bf 0.008} & 0.996 & {\bf 0.004} & {\bf red} \\
GARCH.BU & 47.903 & 5 & 0.254 & 0.873 & 0.127 & yellow \\
GARCH.TD & {\bf 46.721} & 1 & ---- &---- &---- & ---- \\ 
  HAR.GARCH & 46.885 & 4 & 0.871 & 0.565 & 0.435 & yellow \\
  CR.LSE & {\bf 46.721} & 1 & ---- & ---- & ---- & ---- \\
  CR.OPT & {\bf 46.721} & 1 & ---- & ---- & ---- & ---- \\  \hline
    \multicolumn{7}{l}{(2) Total ES}\\
  HS & 423.005 & 6 & {\bf 0.003} & 0.998 & {\bf 0.002} & {\bf red} \\
  GARCH.BU & 393.015 & 5 & 0.498 & 0.751 & 0.249 & yellow \\
  GARCH.TD & {\bf 391.436} & 1 & ---- & ---- & ---- & yellow \\
  HAR.GARCH & 392.767 & 4 & 0.653 & 0.673 & 0.327 & yellow \\
  CR.LSE & {\bf 391.436} & 1 & ---- & ---- & ----& ---- \\
  CR.OPT & {\bf 391.436} & 1 & ---- &  ----& ---- &---- \\
  \hline
  \multicolumn{7}{l}{(3) Tuple of ESCs}\\
HS & 189.216 & 6 & {\bf 0.003} & 0.998 & {\bf 0.001} & ({\bf red}) \\
  GARCH.BU & 118.196 & 4 & {\bf 0.021} & 0.990 & {\bf 0.010} & {\bf red} \\
  GARCH.TD & {\bf 94.699} & 2 & 0.255 & 0.872 & 0.128 & yellow \\
  HAR.GARCH & 124.779 & 5 &  0.202 & 0.899 & 0.101 & yellow \\
  CR.LSE & {\bf 95.256} & 3 & 0.086 & 0.957 & {\bf 0.043} & {\bf red} \\
  CR.OPT & {\bf 85.839} & 1 & ---- & ---- & ---- & ---- \\
  \hline
\end{tabular}
}
\begin{tablenotes}
\item $^a$Average scores are multiplied by $100$. \item $^b$Ranks are based on average scores. \item $^c$The p-values are calculated based on the three different null hypotheses, where ``$=$ CR.OPT" means that the model is equally accurate as CR.OPT, and ``$\le$ ($\ge$) CR.OPT" represents the hypothesis that the model is less (more) accurate than CR.OPT.
\item $^d$Results of the three-zone approach are presented.
For (3), the result of HS is enclosed by parentheses since the order on total VaR is already supported in (1).
\end{tablenotes}
\end{table}

\begin{table}[p]
\caption{Results of the DM-tests to compare the forecast accuracy of $j$th ESC, $j\in\{1,\dots,d\}$, with CR.OPT as the benchmark model.
See the footnotes in~Table~\ref{table:test:var:es:esc} for details.}
\label{table:test:escs}
\centering
\scalebox{1}{
\begin{tabular}{lrrrrrr}
  \hline
 & Average score & Rank &  \multicolumn{3}{c}{p-value} & Region
 \\
  \hline
H0 & &  & $=$ CR.OPT & $\le$ CR.OPT & $\ge$ CR.OPT & \\
  \hline
    \multicolumn{7}{l}{(1) ESC (AMZN)}\\
HS & 33.603 & 6 & 0.091 & 0.954 & {\bf 0.046} & ({\bf red}) \\ 
  GARCH.BU & {\bf 18.646} & 2 & 0.857 & 0.428 & 0.572 & yellow \\ 
  GARCH.TD & 19.951 & 4 & 0.784 & 0.608 & 0.392 & yellow \\ 
  HAR.GARCH & {\bf 13.351} & 1 & 0.495 & 0.247 & 0.753 & yellow \\ 
  CR.LSE & 23.295 & 5 & 0.152 & 0.924 & 0.076 & yellow \\ 
  CR.OPT & {\bf 19.291} & 3 & ---- & ---- & ---- & ---- \\ 
  \hline    \multicolumn{7}{l}{(2) ESC (GOOGL)}\\
  HS & 28.672 & 6 & 0.083 & 0.959 & {\bf 0.041} & ({\bf red}) \\ 
  GARCH.BU & 13.291 & 4 & 0.496 & 0.752 & 0.248 & yellow \\ 
  GARCH.TD & {\bf 11.865} & 3 & 0.963 & 0.518 & 0.482 & yellow \\ 
  HAR.GARCH & 15.629 & 5 & 0.285 & 0.857 & 0.143 & yellow \\ 
  CR.LSE & {\bf 11.054} & 1 & 0.658 & 0.329 & 0.671 & yellow \\ 
  CR.OPT & {\bf 11.782} & 2 & ---- & ---- & ---- & ---- \\ 
   \hline    \multicolumn{7}{l}{(3) ESC (TSLA)}\\
   HS & 126.941 & 6 & {\bf 0.011} & 0.994 & {\bf 0.006} & ({\bf red}) \\ 
  GARCH.BU & 86.260 & 4 & {\bf 0.011} & 0.994 & {\bf 0.006} & {\bf red} \\ 
  GARCH.TD & {\bf 62.882} & 3 & 0.257 & 0.871 & 0.129 & yellow \\ 
  HAR.GARCH & 95.799 & 5 & 0.130 & 0.935 & 0.065 & yellow \\ 
  CR.LSE & {\bf 62.115} & 2 & 0.078 & 0.961 & {\bf 0.039} & {\bf red} \\ 
  CR.OPT & {\bf 54.765} & 1 & ---- & ---- & ---- &----  \\ 
   \hline
\end{tabular}
}
\end{table}

We next describe our traditional backtests.
Denote by $\left(\widehat{ESC}_{1,t},\dots,\widehat{ESC}_{d,t},\widehat{VaR}_t,\widehat{ES}_t\right)$, $t\in \mathbb{T}_{\text{out}}$, the tuple of forecasted risk functionals.
For total VaR, let $\V_t^{\text{VaR}}=\V^\text{VaR}\left(\widehat{VaR}_t,S_t\right)$, where $\V^{\text{VaR}}$ is as defined in Proposition~\ref{prop:joint:identifiability:esc}.
We construct the DM-type test statistic $\sqrt{T}\bar \V_T^{\text{VaR}}/\hat \sigma_T^{\text{VaR}}$, where $\bar \V_T^{\text{VaR}}=(1/T)\sum_{t=1}^T \V_t^{\text{VaR}}$, and $\hat \sigma_T^{\text{VaR}}$ is the HAC estimator of the standard deviation of $\sqrt{T}\bar \V_T^{\text{VaR}}$ as considered in the comparative backtests above.
Following Section~2.2 of~\citet{nolde2017elicitability}, we use a normal distribution as a null distribution to test \emph{unconditional calibration}, \emph{sub-calibration}, and \emph{super-calibration}, which are associated with precise estimation, over-estimation, and under-estimation, respectively, for this identification function of VaR; see Section~\ref{sec:order:sensitivity}.
Following these relationships, we adopt the three-zone approach and say that the forecast model is in the red region if the null hypothesis of sub-calibration is rejected, in the green region if the null hypothesis of super-calibration is rejected, and in the yellow region if none of these hypotheses are rejected.
For total ES, we replace $\V_t^{\text{VaR}}$ in the above analysis with $\V_t^\text{ES}=\V^\text{ES}\left((\widehat{VaR}_t,\widehat{ES}_t),S_t\right)$, where:
\begin{align*}
\V^\text{ES}\left((v,e),s\right)=v-e-\frac{1}{1-\alpha}\id\{s>v\}(v-s).
\end{align*}

This identification function maintains the above relationships on over- and under-estimation; see Section~2.2.2 of~\citet{nolde2017elicitability}.
For $j$th ESC, $j\in \{1,\dots,d\}$, we replace $\V_t^{\text{VaR}}$ in the above analysis with 
\begin{align*}
\V_{j,t}^\text{ESC}=\V_j^\text{ESC}\left((\widehat{ESC}_{j,t},\widehat{VaR}_t),\bX_t\right),
\end{align*}
where $\V_j^\text{ESC}$ is as defined in Proposition~\ref{prop:joint:identifiability:esc}.
We report the results of these tests in Tables~\ref{table:test:V:var:es} and~\ref{table:test:V:escs}.
Due to joint identifiability of total ES~\citep[see~Section~2.1~of~][]{nolde2017elicitability} and $j$th ESC for $j=1,\dots,d$
 (see Proposition~\ref{prop:joint:identifiability:esc}) in combination with total VaR, we also conduct Wald-tests for these risk quantities in the one-step approach; see Section~\ref{sec:one:step:approach} for details.

\begin{table}[p]
\caption{
Results of the DM-type tests to verify the forecast accuracy of total VaR and total ES.
}
\label{table:test:V:var:es}
\centering
\scalebox{1}{
\begin{tabular}{lrrrrr}
  \hline
 & Average score & \multicolumn{3}{c}{p-value$^a$} & Region$^b$
 \\
  \hline
H0 & & $=$ True & $\le$ True & $\ge$ True & \\ \hline
  \multicolumn{6}{l}{(1) Total VaR}\\

HS & 0.032 & {\bf 0.000} & 1.000 & {\bf 0.000} & {\bf red} \\ 
  GARCH.BU & 0.016 & {\bf 0.009} & 0.995 & {\bf 0.005} & {\bf red} \\ 
  GARCH.TD & {\bf 0.009} & 0.116 & 0.942 & 0.058 & yellow \\ 
  HAR.GARCH & 0.011 & 0.057 & 0.971 & {\bf 0.029} & {\bf red} \\ 
  CR.LSE & {\bf 0.009} & 0.116 & 0.942 & 0.058 & yellow \\ 
  CR.OPT & {\bf 0.009} & 0.116 & 0.942 & 0.058 & yellow \\ 
  \hline
    \multicolumn{6}{l}{(2) Total ES}\\
HS & 6.441 & {\bf 0.002} & 0.999 & {\bf 0.001} & {\bf red} \\ 
  GARCH.BU & 0.634 & 0.645 & 0.677 & 0.323 & yellow \\ 
  GARCH.TD & {\bf 0.508} & 0.652 & 0.674 & 0.326 & yellow \\ 
  HAR.GARCH & 1.768 & 0.169 & 0.916 & 0.084 & yellow \\ 
  CR.LSE & {\bf 0.508} & 0.652 & 0.674 & 0.326 & yellow \\ 
  CR.OPT & {\bf 0.508} & 0.652 & 0.674 & 0.326 & yellow \\
   \hline
\end{tabular}
}
\begin{tablenotes}
    \item $^a$The p-values are calculated based on the three different null hypotheses, where ``$=$ True'' means that the forecast is precise, ``$\le$ True'' represents the hypothesis that the forecast is under-estimated, and ``$\ge$ True'' stands for the case when the forecast is over-estimated.
    \item $^b$ Results of the three-zone approach are presented.
\end{tablenotes}
\end{table}

\begin{table}[p]
\caption{
Results of the DM-type tests to verify the forecast accuracy of the $j$th ESC for $j=1,\dots,d$.
See the footnotes in Table~\ref{table:test:V:var:es} for details.
}
\label{table:test:V:escs}
\centering
\scalebox{1}{
\begin{tabular}{lrrrrr}
  \hline
  & Average score & \multicolumn{3}{c}{p-value} & Region
 \\
  \hline
H0 & & $=$ True & $\le$ True & $\ge$ True & \\
\hline
  \multicolumn{6}{l}{(1) ESC (AMZN)}\\
   HS & 0.013 & 0.487 & 0.757 & 0.243 & yellow \\ 
  GARCH.BU & {\bf -0.005} & 0.700 & 0.350 & 0.650 & yellow \\ 
  GARCH.TD & {\bf -0.008} & 0.594 & 0.297 & 0.703 & yellow \\ 
  HAR.GARCH & 0.032 & {\bf 0.005} & 0.998 & {\bf 0.002} & {\bf red} \\ 
  CR.LSE & {\bf -0.007} & 0.656 & 0.328 & 0.672 & yellow \\ 
  CR.OPT & -0.014 & 0.296 & 0.148 & 0.852 & yellow \\ 
  \hline    \multicolumn{6}{l}{(2) ESC (GOOGL)}\\
  HS & 0.029 & 0.088 & 0.956 & {\bf 0.044} & {\bf red} \\ 
  GARCH.BU & {\bf -0.005} & 0.669 & 0.335 & 0.665 & yellow \\ 
  GARCH.TD & -0.009 & 0.393 & 0.196 & 0.804 & yellow \\ 
  HAR.GARCH & 0.043 & {\bf 0.001} & 1.000 & {\bf 0.000} & {\bf red} \\ 
  CR.LSE & {\bf -0.006} & 0.550 & 0.275 & 0.725 & yellow \\ 
  CR.OPT & {\bf -0.005} & 0.643 & 0.322 & 0.678 & yellow \\
   \hline    \multicolumn{6}{l}{(3) ESC (TSLA)}\\
   HS & {\bf -0.009} & 0.803 & 0.401 & 0.599 & yellow \\ 
  GARCH.BU & -0.070 & {\bf 0.017} & {\bf 0.009} & 0.991 & {\bf green} \\ 
  GARCH.TD & -0.023 & 0.353 & 0.177 & 0.823 & yellow \\ 
  HAR.GARCH & 0.105 & {\bf 0.001} & 1.000 & {\bf 0.000} & {\bf red} \\ 
  CR.LSE & {\bf -0.015} & 0.536 & 0.268 & 0.732 & yellow \\ 
  CR.OPT & {\bf -0.009} & 0.701 & 0.351 & 0.649 & yellow \\ 
   \hline
\end{tabular}
}
\end{table}

In the end, we visually check the performance of the forecasts by Murphy diagrams.
For total VaR, the diagram is written based on~\eqref{eq:var:elementary}; see~\citet{ehm2016quantiles} for details.
We refer to~\citet{ziegel2020robust} for the Murphy diagram of ES.
The Murphy diagram of the tuple of ESCs is drawn based on Proposition~\ref{prop:esc:mixture} (M2).
We also display Murphy diagrams of the $j$th ESC for $j=1,\dots,d$ based on Proposition~\ref{prop:esc:mixture} (M1).
The results appear in Figure~\ref{fig:murphy}.

\begin{figure}[t]
\begin{tabular}{cc}
      \begin{minipage}[t]{0.48\hsize}
        \centering
\includegraphics[width=1.0\textwidth]{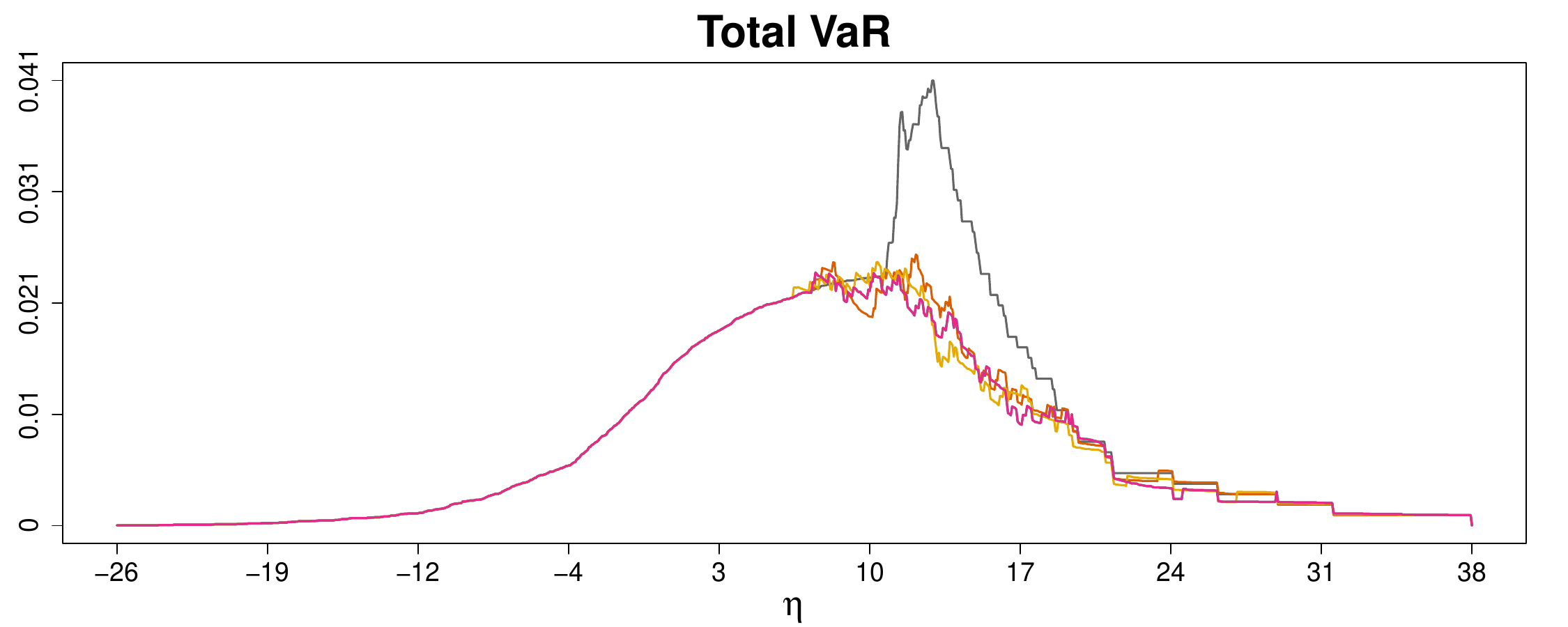}
 \end{minipage} &
       \begin{minipage}[t]{0.48\hsize}
        \centering
\includegraphics[width=1.0\textwidth]{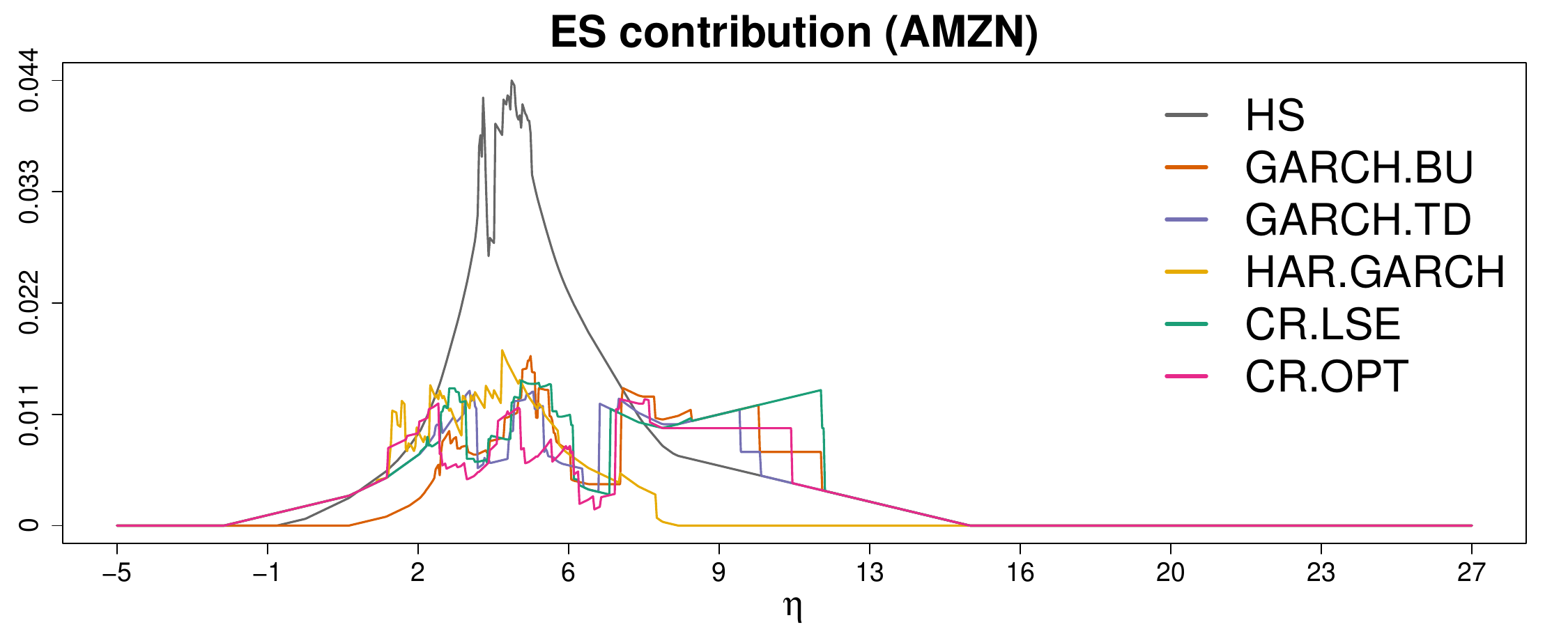}
 \end{minipage}\\
 \begin{minipage}[t]{0.48\hsize}
        \centering
\includegraphics[width=1.0\textwidth]{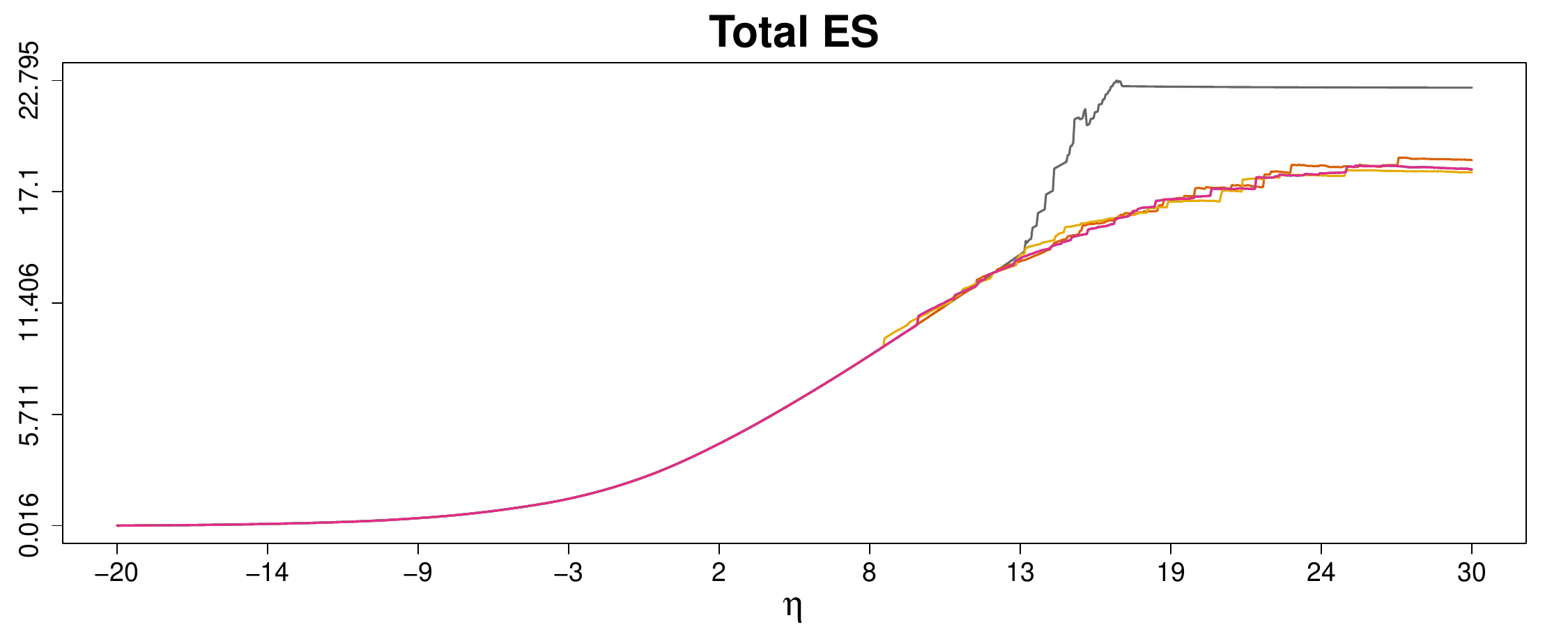}
 \end{minipage} &
       \begin{minipage}[t]{0.48\hsize}
        \centering
\includegraphics[width=1.0\textwidth]{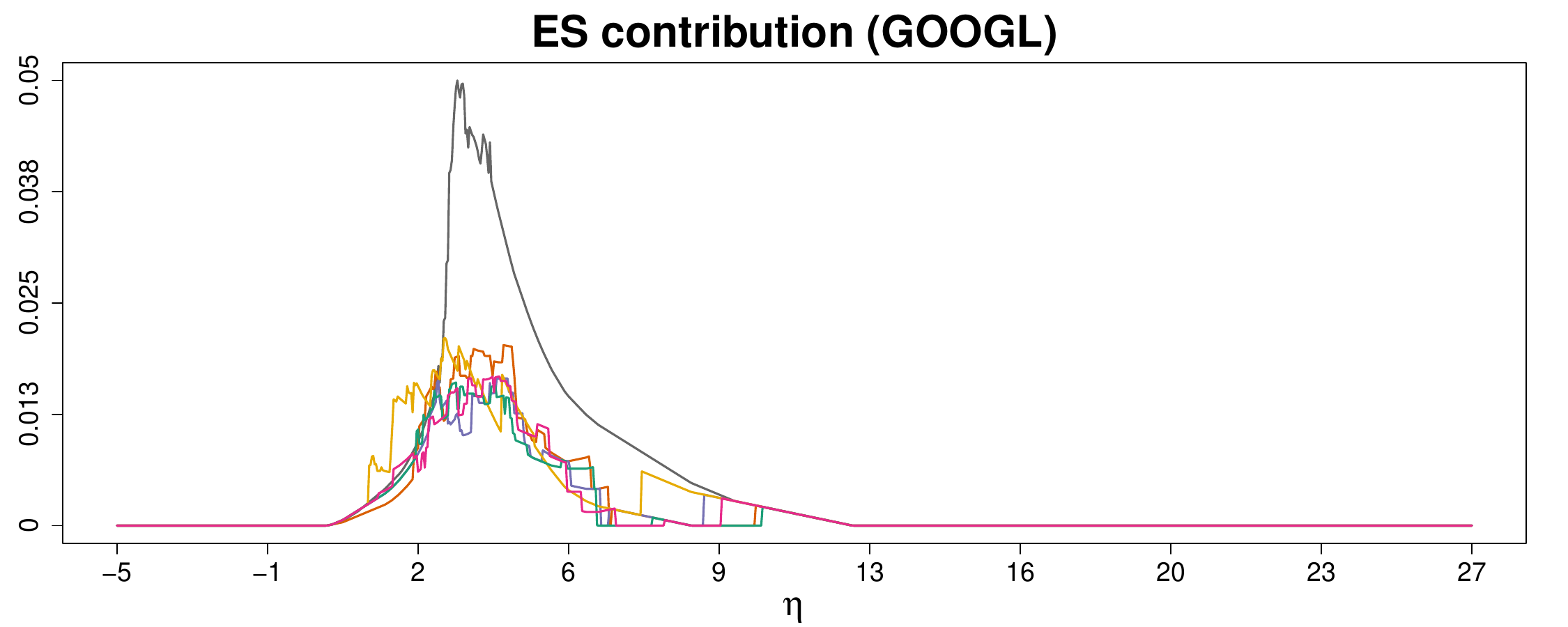}
 \end{minipage}\\
 \begin{minipage}[t]{0.48\hsize}
        \centering
\includegraphics[width=1.0\textwidth]{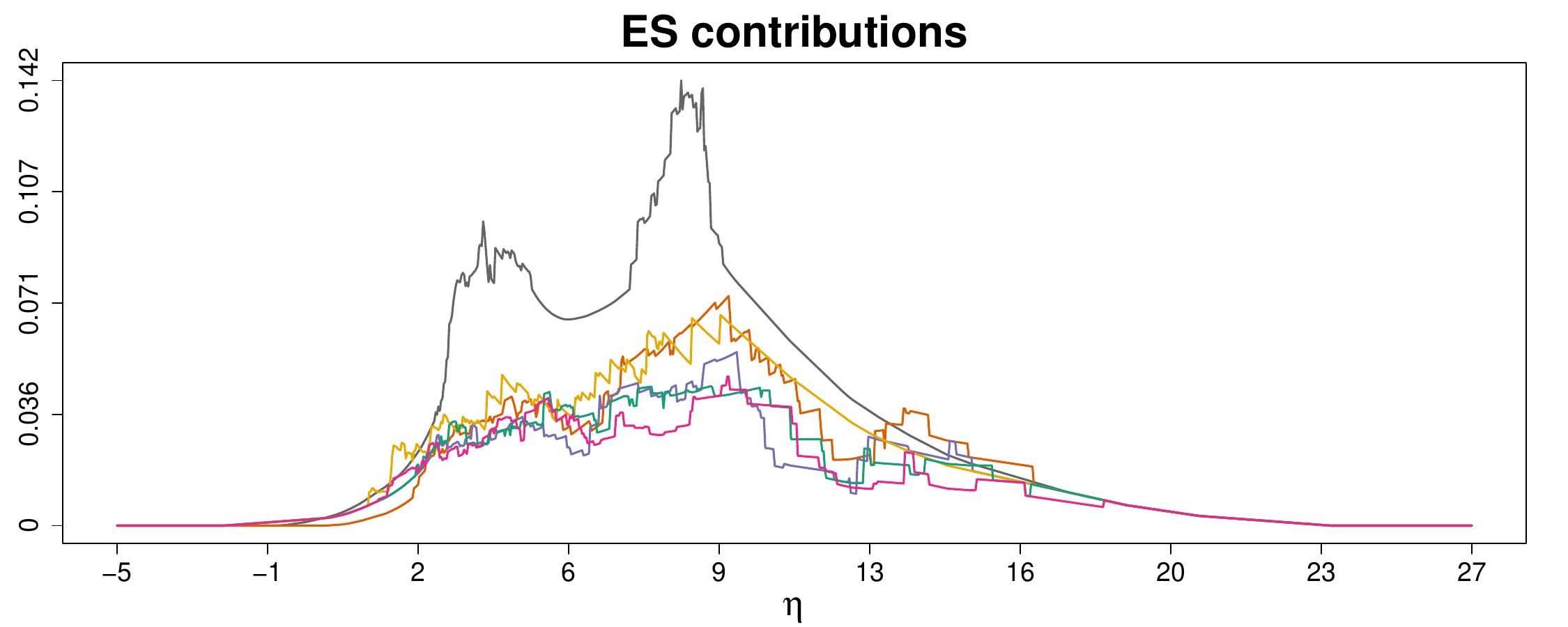}
 \end{minipage} &
       \begin{minipage}[t]{0.48\hsize}
        \centering
\includegraphics[width=1.0\textwidth]{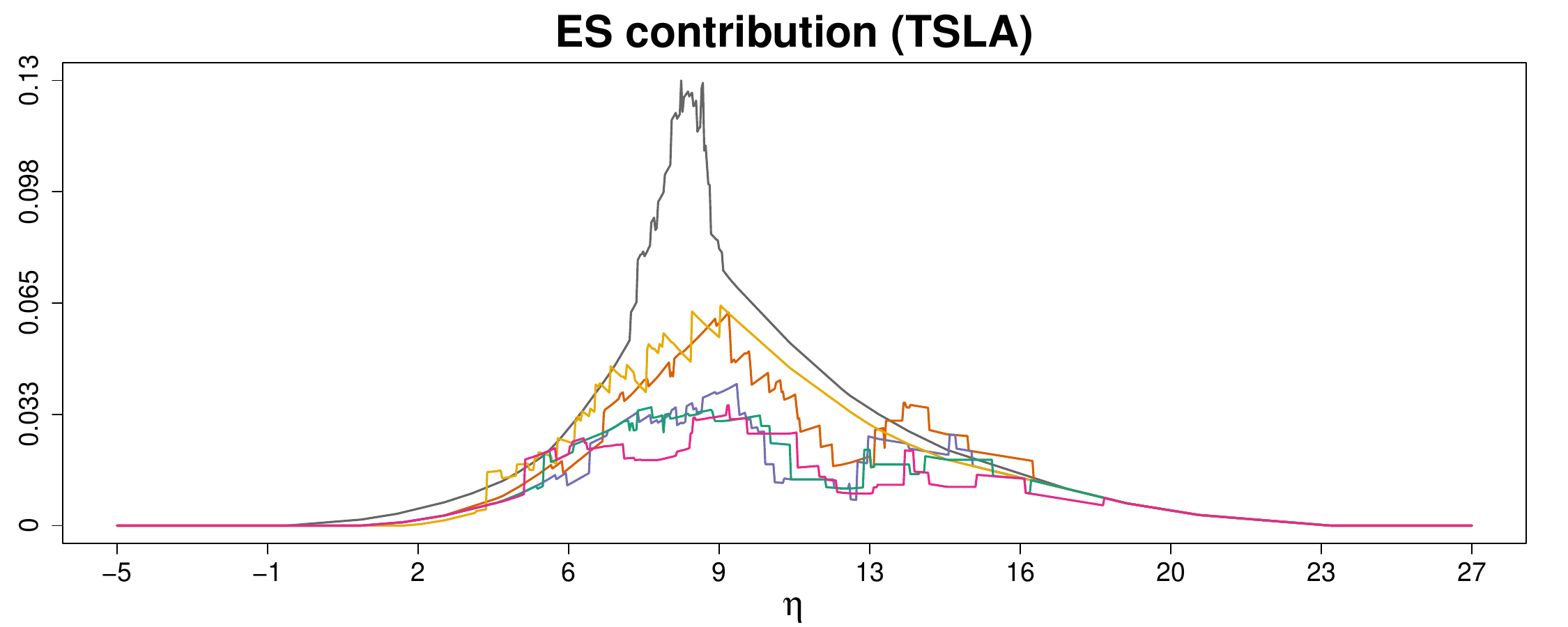}
 \end{minipage}\\
 \end{tabular}
\caption{Murphy diagrams of total VaR, total ES, tuple of ESCs, and each of them with confidence level of $0.975$ estimated by the six models.}
\label{fig:murphy}
\end{figure}

\subsection{Discussion}\label{sec:discussion}

According to the results of the comparative backtests in Tables~\ref{table:test:var:es:esc} and~\ref{table:test:escs}, we first observe that HS is inferior to others in terms of average scores, which is detected by the red region in the three-zone approach.
Moreover, CR.OPT tends to have lower scores than others and achieves the best performance for the tuple of ESCs and the ESC of TSLA.
Overall, GARCH.TD outperforms GARCH.BU, and HAR.GARCH does not perform well except for the ESC of AMZN.
For total VaR in Table~\ref{table:test:var:es:esc} (1), we observe the tendency that the forecasts of total VaR used in GARCH.TD, CR.LSE and CR.OPT are more accurate than others.
Nevertheless, no test statistically supports the difference in the forecast accuracy of total VaRs except for HS, which justifies the forecast comparisons of ESCs in terms of their scores.

In Table~\ref{table:test:var:es:esc} (2) for total ES, we find slightly lower performances of GARCH.BU and HAR.GARCH than others.
This indicates that the top-down approach can be slightly more preferable than the bottom-up approach, and it may not be recommendable to estimate total ES indirectly as the sum of ESCs.
In addition to total ES, both of GARCH.BU and HAR.GARCH do not perform well for the tuple of ESCs as well as each of them except for AMZN.
For the case of AMZN, HAR.GARCH performs the best, which may imply the existence of the hysteretic effect and/or dynamic conditional correlations in the series of stock returns.
From the perspective of multi-objective elicitability, we can directly compare the forecast accuracy of ESCs among GARCH.TD, CR.LSE and CR.OPT since they share common forecasts of total VaR.
We observe that CR.OPT outperforms CR.LSE for all the cases except the ESC of GOOGL.
Finally, we see in Table~\ref{table:S:test:wald} of Section~\ref{sec:one:step:approach} that the Wald-tests statistically support the superiority of CR.OPT to HS for AMZN and GOOGL and to HS and GARCH.BU for the tuple of ESCs and TSLA.
These orders are also detected in Table~\ref{table:test:var:es:esc}.
In addition, the Wald-tests do not support the order between CR.LSE and CR.OPT, which is supported in the DM-tests in Tables~\ref{table:test:var:es:esc} and~\ref{table:test:escs}.

We next check the results of the traditional backtests in Tables~\ref{table:test:V:var:es} and~\ref{table:test:V:escs}. Overall, the results are consistent with those of the comparative backtests, and HS particularly does not identify the true risk quantities well.
Other than HS, we find that GARCH.BU and HAR.GARCH are two models that typically fail the tests.
For these models, the forecasts are often under-estimated except for the case that GARCH.BU over-estimates the ESC of TSLA.
The other three models, GARCH.TD, CR.LSE, and CR.OPT, pass the traditional backtests.
These results are consistent with those in Table~\ref{table:V:test:wald} of Section~\ref{sec:one:step:approach}, where HS, GARCH.BU and HAR.GARCH typically fail the model validation tests.

Since all the above results depend on the specific choices of scoring functions, it is beneficial to check the Murphy diagrams to diagnose the robustness of our observations for different choices of scoring functions.
Overall, we do not observe clear uniform dominance among the curves except HS, which is distinguishable for all the risk quantities.
Compared with total VaR and total ES, larger differences are more visible for ESCs.
Particularly for the tuple of ESCs and the ESC of TSLA, the curves of CR.OPT are typically lower, and those of GARCH.BU and HAR.GARCH tend to be higher.
In addition, the curves are more tangled for ESCs of AMZN and GOOGL.
These differences are consistent with the fact that, from Table~\ref{table:test:escs} and the y-axes of Murphy diagrams of ESCs, TSLA has larger contribution to the score of the tuple of ESCs than AMZN and GOOGL have, and thus optimization in CR.OPT should put higher weight to the score of TSLA under FAP.
In summary, the presented diagrams visually indicate the superiority of CR.OPT to other models for the tuple of ESCs and the ESC of TSLA.
For the ESC of AMZN, and possibly of GOOGL, ranking of models may vary depending on the choice of the scoring function.

\section{Conclusion and outlook}\label{sec:conclusion:outlook}

We conduct traditional and comparative backtests for the gradient allocations of ES, and visually check the robustness of our observations based on Murphy diagrams.
Moreover, motivated by the comparability of ESCs in multi-objective elicitability and by the preservability of FAP, we propose a novel semiparametric compositional regression model to estimate the tuple of ESCs.
Our empirical analysis demonstrates superior performance of our proposed model for forecasting ESCs.
We believe that these results are of great benefit for portfolio and enterprise risk management in financial and regulatory applications.

We conclude this section by offering an outlook for future research on backtesting ESCs and our proposed model. First, exploring theoretical aspects of our proposed compositional regression model, such as stability, consistency, and asymptotic normality, is of great interest.
As a practical aspect, variable selection for our proposed model and the choice of the transform between the simplex and the real space are beyond the scope of the present paper.
The cases of high-dimensional portfolios and of negative or zero ESCs may require further analyses.
Moreover, in a future work it may be interesting to extend our estimation procedure of risk allocations to VaR contributions, the gradient allocation of VaR~\citep[see][for recent works]{koike2022avoiding,gribkova2023estimating}, and other capital allocation rules based on optimization \citep{dhaene2012optimal,maume2016capital,koike2021modality}.
Finally, other backtesting procedures~\citep{banulescu2021backtesting,wang2022backtesting,hoga2023monitoring,hue2024backtesting} could bring new insights on backtesting capital allocations.

\section*{Acknowledgements}\label{sec:acknowledgements}
Takaaki Koike was supported by Japan Society for the Promotion of Science (JSPS KAKENHI Grant Numbers JP21K13275 and JP24K00273). Cathy W.S. Chen's research is funded by the National Science and Technology Council, Taiwan (NSTC112-2118-M-035-001-MY3).
Edward M.H. Lin’s research is funded by the National Science and Technology Council, Taiwan (NSTC111-2118-M-029-003-MY2).
The authors also thank Tim J. Boonen for constructive comments and suggestions.


\section*{Declaration of interest}
The authors declare that they have no known competing financial interests or personal relationships that could have appeared to influence the work reported in this paper.


\begin{thebibliography}{}

\bibitem[\protect\citeauthoryear{Aitchison}{Aitchison}{1982}]{aitchison1982statistical}
Aitchison, J. (1982).
\newblock The statistical analysis of compositional data.
\newblock {\em Journal of the Royal Statistical Society: Series B
  (Methodological)\/}~{\em 44\/}(2), 139--160.

\bibitem[\protect\citeauthoryear{Aitchison and J.~Egozcue}{Aitchison and
  J.~Egozcue}{2005}]{aitchison2005compositional}
Aitchison, J. and J.~J.~Egozcue (2005).
\newblock Compositional data analysis: where are we and where should we be
  heading?
\newblock {\em Mathematical Geology\/}~{\em 37}, 829--850.

\bibitem[\protect\citeauthoryear{Andrews}{Andrews}{1991}]{andrews1991heteroskedasticity}
Andrews, D.~W. (1991).
\newblock Heteroskedasticity and autocorrelation consistent covariance matrix.
\newblock {\em Econometrica\/}~{\em 59\/}(3), 817--858.

\bibitem[\protect\citeauthoryear{Banulescu-Radu, Hurlin, Leymarie, and
  Scaillet}{Banulescu-Radu et~al.}{2021}]{banulescu2021backtesting}
Banulescu-Radu, D., C.~Hurlin, J.~Leymarie, and O.~Scaillet (2021).
\newblock Backtesting marginal expected shortfall and related systemic risk
  measures.
\newblock {\em Management science\/}~{\em 67\/}(9), 5730--5754.

\bibitem[\protect\citeauthoryear{BCBS}{BCBS}{2013}]{bcbs2013consultative}
BCBS (2013).
\newblock Consultative document {O}ctober 2013. fundamental review of the
  trading book: A revised market risk framework.
\newblock Basel Committee on Banking Supervision. Basel: Bank for International
  Settlements. BIS online publication. No. bcbs265.

\bibitem[\protect\citeauthoryear{BCBS}{BCBS}{2016}]{bcbs2016minimum}
BCBS (2016).
\newblock Minimum capital requirements for market risk. {J}anuary 2016.
\newblock Basel Committee on Banking Supervision. Basel: Bank for International
  Settlements. BIS online publication. No. d352.

\bibitem[\protect\citeauthoryear{BCBS}{BCBS}{2019}]{bcbs2019minimum}
BCBS (2019).
\newblock Minimum capital requirements for market risk. {F}ebruary 2019.
\newblock Basel Committee on Banking Supervision. Basel: Bank for International
  Settlements. BIS online publication. No. d457.

\bibitem[\protect\citeauthoryear{Bielecki, Cialenco, Pitera, and
  Schmidt}{Bielecki et~al.}{2020}]{bielecki2020fair}
Bielecki, T.~R., I.~Cialenco, M.~Pitera, and T.~Schmidt (2020).
\newblock Fair estimation of capital risk allocation.
\newblock {\em Statistics \& Risk Modeling\/}~{\em 37\/}(1-2), 1--24.

\bibitem[\protect\citeauthoryear{Boonen, Guillen, and Santolino}{Boonen
  et~al.}{2019}]{boonen2019forecasting}
Boonen, T.~J., M.~Guillen, and M.~Santolino (2019).
\newblock Forecasting compositional risk allocations.
\newblock {\em Insurance: Mathematics and Economics\/}~{\em 84}, 79--86.

\bibitem[\protect\citeauthoryear{Chen, Than-Thi, So, and Sriboonchitta}{Chen
  et~al.}{2019}]{chen2019quantile}
Chen, C.~W., H.~Than-Thi, M.~K. So, and S.~Sriboonchitta (2019).
\newblock Quantile forecasting based on a bivariate hysteretic autoregressive
  model with garch errors and time-varying correlations.
\newblock {\em Applied Stochastic Models in Business and Industry\/}~{\em
  35\/}(6), 1301--1321.

\bibitem[\protect\citeauthoryear{Demarta and McNeil}{Demarta and
  McNeil}{2005}]{demarta2005t}
Demarta, S. and A.~J. McNeil (2005).
\newblock The t copula and related copulas.
\newblock {\em International statistical review\/}~{\em 73\/}(1), 111--129.

\bibitem[\protect\citeauthoryear{Denault}{Denault}{2001}]{denault2001coherent}
Denault, M. (2001).
\newblock Coherent allocation of risk capital.
\newblock {\em Journal of Risk\/}~{\em 4\/}(1), 1--34.

\bibitem[\protect\citeauthoryear{Dhaene, Tsanakas, Valdez, and
  Vanduffel}{Dhaene et~al.}{2012}]{dhaene2012optimal}
Dhaene, J., A.~Tsanakas, E.~A. Valdez, and S.~Vanduffel (2012).
\newblock Optimal capital allocation principles.
\newblock {\em Journal of Risk and Insurance\/}~{\em 79\/}(1), 1--28.

\bibitem[\protect\citeauthoryear{Diebold and Mariano}{Diebold and
  Mariano}{1995}]{diebold1995comparing}
Diebold, F.~X. and R.~S. Mariano (1995).
\newblock Comparing predictive accuracy.
\newblock {\em Journal of Business \& Economic Statistics\/}~{\em 13\/}(3),
  134--144.

\bibitem[\protect\citeauthoryear{Dimitriadis and Hoga}{Dimitriadis and
  Hoga}{2023}]{dimitriadis2023dynamic}
Dimitriadis, T. and Y.~Hoga (2023).
\newblock Dynamic covar modeling.
\newblock {\em arXiv preprint, 2206.14275\/}.

\bibitem[\protect\citeauthoryear{Egozcue, Pawlowsky-Glahn, Mateu-Figueras, and
  Barcelo-Vidal}{Egozcue et~al.}{2003}]{egozcue2003isometric}
Egozcue, J.~J., V.~Pawlowsky-Glahn, G.~Mateu-Figueras, and C.~Barcelo-Vidal
  (2003).
\newblock Isometric logratio transformations for compositional data analysis.
\newblock {\em Mathematical Geology\/}~{\em 35\/}(3), 279--300.

\bibitem[\protect\citeauthoryear{Ehm, Gneiting, Jordan, and Kr{\"u}ger}{Ehm
  et~al.}{2016}]{ehm2016quantiles}
Ehm, W., T.~Gneiting, A.~Jordan, and F.~Kr{\"u}ger (2016).
\newblock Of quantiles and expectiles: consistent scoring functions, choquet
  representations and forecast rankings.
\newblock {\em Journal of the Royal Statistical Society Series B: Statistical
  Methodology\/}~{\em 78\/}(3), 505--562.

\bibitem[\protect\citeauthoryear{Emmer, Kratz, and Tasche}{Emmer
  et~al.}{2015}]{emmer2015best}
Emmer, S., M.~Kratz, and D.~Tasche (2015).
\newblock What is the best risk measure in practice? a comparison of standard
  measures.
\newblock {\em Journal of Risk\/}~{\em 18\/}(2), 31--60.

\bibitem[\protect\citeauthoryear{Fern{\'a}ndez and Steel}{Fern{\'a}ndez and
  Steel}{1998}]{fernandez1998bayesian}
Fern{\'a}ndez, C. and M.~F. Steel (1998).
\newblock On bayesian modeling of fat tails and skewness.
\newblock {\em Journal of the American Statistical Association\/}~{\em
  93\/}(441), 359--371.

\bibitem[\protect\citeauthoryear{Fissler and Hoga}{Fissler and
  Hoga}{2024}]{fissler2024backtesting}
Fissler, T. and Y.~Hoga (2024).
\newblock Backtesting systemic risk forecasts using multi-objective
  elicitability.
\newblock {\em Journal of Business \& Economic Statistics\/}~{\em 42\/}(2),
  485--498.

\bibitem[\protect\citeauthoryear{Fissler and Ziegel}{Fissler and
  Ziegel}{2016}]{fissler2016higher}
Fissler, T. and J.~F. Ziegel (2016, Aug).
\newblock Higher order elicitability and osband’s principle.
\newblock {\em The Annals of Statistics\/}~{\em 44\/}(4), 1680--1707.

\bibitem[\protect\citeauthoryear{Fissler and Ziegel}{Fissler and
  Ziegel}{2019}]{fissler2019order}
Fissler, T. and J.~F. Ziegel (2019).
\newblock Order-sensitivity and equivariance of scoring functions.
\newblock {\em Electronic Journal of Statistics\/}~{\em 13}, 1166--1211.

\bibitem[\protect\citeauthoryear{Fissler, Ziegel, and Gneiting}{Fissler
  et~al.}{2016}]{fissler2016expected}
Fissler, T., J.~F. Ziegel, and T.~Gneiting (2016, January).
\newblock Expected shortfall is jointly elicitable with value at
  risk-implications for backtesting.
\newblock {\em Risk Magazine\/}, 58--61.

\bibitem[\protect\citeauthoryear{Gribkova, Su, and Zitikis}{Gribkova
  et~al.}{2023}]{gribkova2023estimating}
Gribkova, N., J.~Su, and R.~Zitikis (2023).
\newblock Estimating the var-induced euler allocation rule.
\newblock {\em ASTIN Bulletin: The Journal of the IAA\/}~{\em 53\/}(3),
  619--635.

\bibitem[\protect\citeauthoryear{Hofert, Kojadinovic, Maechler, and Yan}{Hofert
  et~al.}{2023}]{copula}
Hofert, M., I.~Kojadinovic, M.~Maechler, and J.~Yan (2023).
\newblock {\em copula: Multivariate Dependence with Copulas}.
\newblock R package version 1.1-2.

\bibitem[\protect\citeauthoryear{Hoga and Demetrescu}{Hoga and
  Demetrescu}{2023}]{hoga2023monitoring}
Hoga, Y. and M.~Demetrescu (2023).
\newblock Monitoring value-at-risk and expected shortfall forecasts.
\newblock {\em Management Science\/}~{\em 69\/}(5), 2954--2971.

\bibitem[\protect\citeauthoryear{Huang, Lee, Liang, and Lin}{Huang
  et~al.}{2009}]{huang2009estimating}
Huang, J.-J., K.-J. Lee, H.~Liang, and W.-F. Lin (2009).
\newblock Estimating value at risk of portfolio by conditional copula-garch
  method.
\newblock {\em Insurance: Mathematics and Economics\/}~{\em 45\/}(3), 315--324.

\bibitem[\protect\citeauthoryear{Hu{\'e}, Hurlin, and Lu}{Hu{\'e}
  et~al.}{2024}]{hue2024backtesting}
Hu{\'e}, S., C.~Hurlin, and Y.~Lu (2024).
\newblock Backtesting expected shortfall: Accounting for both duration and
  severity with bivariate orthogonal polynomials.
\newblock {\em arXiv preprint arXiv:2405.02012\/}.

\bibitem[\protect\citeauthoryear{Jondeau and Rockinger}{Jondeau and
  Rockinger}{2006}]{jondeau2006copula}
Jondeau, E. and M.~Rockinger (2006).
\newblock The copula-garch model of conditional dependencies: An international
  stock market application.
\newblock {\em Journal of International Money and Finance\/}~{\em 25\/}(5),
  827--853.

\bibitem[\protect\citeauthoryear{Kalkbrener}{Kalkbrener}{2005}]{kalkbrener2005axiomatic}
Kalkbrener, M. (2005).
\newblock An axiomatic approach to capital allocation.
\newblock {\em Mathematical Finance\/}~{\em 15\/}(3), 425--437.

\bibitem[\protect\citeauthoryear{Koike and Hofert}{Koike and
  Hofert}{2021}]{koike2021modality}
Koike, T. and M.~Hofert (2021).
\newblock Modality for scenario analysis and maximum likelihood allocation.
\newblock {\em Insurance: Mathematics and Economics\/}~{\em 97}, 24--43.

\bibitem[\protect\citeauthoryear{Koike, Saporito, and Targino}{Koike
  et~al.}{2022}]{koike2022avoiding}
Koike, T., Y.~Saporito, and R.~Targino (2022).
\newblock Avoiding zero probability events when computing value at risk
  contributions.
\newblock {\em Insurance: Mathematics and Economics\/}~{\em 106}, 173--192.

\bibitem[\protect\citeauthoryear{Maume-Deschamps, Rulli{\`e}re, and
  Said}{Maume-Deschamps et~al.}{2016}]{maume2016capital}
Maume-Deschamps, V., D.~Rulli{\`e}re, and K.~Said (2016).
\newblock On a capital allocation by minimization of some risk indicators.
\newblock {\em European Actuarial Journal\/}~{\em 6\/}(1), 177--196.

\bibitem[\protect\citeauthoryear{McNeil, Frey, and Embrechts}{McNeil
  et~al.}{2015}]{mcneil2015quantitative}
McNeil, A.~J., R.~Frey, and P.~Embrechts (2015).
\newblock {\em Quantitative risk management: Concepts, techniques and tools}.
\newblock Princeton: Princeton University Press.

\bibitem[\protect\citeauthoryear{Newey and West}{Newey and
  West}{1987}]{newey1987simple}
Newey, W.~K. and K.~D. West (1987).
\newblock A simple, positive semi-definite, heteroskedasticity and
  autocorrelation consistent covariance matrix.
\newblock {\em Econometrica\/}~{\em 55\/}(3), 703.

\bibitem[\protect\citeauthoryear{Nolde and Ziegel}{Nolde and
  Ziegel}{2017}]{nolde2017elicitability}
Nolde, N. and J.~F. Ziegel (2017).
\newblock Elicitability and backtesting: Perspectives for banking regulation.
\newblock {\em The annals of Applied Statistics\/}~{\em 11\/}(4), 1833--1874.

\bibitem[\protect\citeauthoryear{Patton}{Patton}{2020}]{patton2020comparing}
Patton, A.~J. (2020).
\newblock Comparing possibly misspecified forecasts.
\newblock {\em Journal of Business \& Economic Statistics\/}~{\em 38\/}(4),
  796--809.

\bibitem[\protect\citeauthoryear{Patton, Ziegel, and Chen}{Patton
  et~al.}{2019}]{patton2019dynamic}
Patton, A.~J., J.~F. Ziegel, and R.~Chen (2019).
\newblock Dynamic semiparametric models for expected shortfall (and
  value-at-risk).
\newblock {\em Journal of Econometrics\/}~{\em 211\/}(2), 388--413.

\bibitem[\protect\citeauthoryear{Pawlowsky-Glahn and Buccianti}{Pawlowsky-Glahn
  and Buccianti}{2011}]{pawlowsky2011compositional}
Pawlowsky-Glahn, V. and A.~Buccianti (2011).
\newblock {\em Compositional Data Analysis: Theory and Applications}.
\newblock John Wiley \& Sons.

\bibitem[\protect\citeauthoryear{{R Core Team}}{{R Core Team}}{2023}]{Rcore}
{R Core Team} (2023).
\newblock {\em R: A Language and Environment for Statistical Computing}.
\newblock Vienna, Austria: R Foundation for Statistical Computing.

\bibitem[\protect\citeauthoryear{Steinwart, Pasin, Williamson, and
  Zhang}{Steinwart et~al.}{2014}]{steinwart2014elicitation}
Steinwart, I., C.~Pasin, R.~Williamson, and S.~Zhang (2014).
\newblock Elicitation and identification of properties.
\newblock In {\em Conference on Learning Theory}, pp.\  482--526. PMLR.

\bibitem[\protect\citeauthoryear{Tasche}{Tasche}{1999}]{tasche1999risk}
Tasche, D. (1999).
\newblock Risk contributions and performance measurement.
\newblock Working Paper, Techische Universit\"at M{\"u}nchen.

\bibitem[\protect\citeauthoryear{Tasche}{Tasche}{2008}]{tasche2008capital}
Tasche, D. (2008).
\newblock Capital allocation to business units and sub-portfolios: the euler
  principle.
\newblock In A.~Resti (Ed.), {\em Pillar II in the New Basel Accord: The
  Challenge of Economic Capital}, pp.\  423--453. Risk Books: London.

\bibitem[\protect\citeauthoryear{Taylor}{Taylor}{2019}]{taylor2019forecasting}
Taylor, J.~W. (2019).
\newblock Forecasting value at risk and expected shortfall using a
  semiparametric approach based on the asymmetric laplace distribution.
\newblock {\em Journal of Business \& Economic Statistics\/}~{\em 37\/}(1),
  121--133.

\bibitem[\protect\citeauthoryear{Taylor}{Taylor}{2022}]{taylor2022forecasting}
Taylor, J.~W. (2022).
\newblock Forecasting value at risk and expected shortfall using a model with a
  dynamic omega ratio.
\newblock {\em Journal of Banking \& Finance\/}~{\em 140}, 106519.

\bibitem[\protect\citeauthoryear{Tse and Tsui}{Tse and Tsui}{2002}]{tse2002}
Tse, Y. and A.~Tsui (2002).
\newblock A {Multivariate} {Generalized} {Autoregressive} {Conditional}
  {Heteroscedasticity} {Model} {With} {Time}-{Varying} {Correlations}.
\newblock {\em Journal of Business \& Economic Statistics\/}~{\em 20\/}(3),
  351--362.

\bibitem[\protect\citeauthoryear{Wang, Wang, and Ziegel}{Wang
  et~al.}{2022}]{wang2022backtesting}
Wang, Q., R.~Wang, and J.~Ziegel (2022).
\newblock E-backtesting.
\newblock {\em arXiv preprint arXiv:2209.00991\/}.

\bibitem[\protect\citeauthoryear{Wuertz, Chalabi, Setz, and Maechler}{Wuertz
  et~al.}{2022}]{fgarch}
Wuertz, D., Y.~Chalabi, T.~Setz, and M.~Maechler (2022).
\newblock {\em fGarch: Rmetrics - Autoregressive Conditional Heteroskedastic
  Modelling}.
\newblock R package version 4022.89.

\bibitem[\protect\citeauthoryear{Ziegel, Kr{\"u}ger, Jordan, and
  Fasciati}{Ziegel et~al.}{2020}]{ziegel2020robust}
Ziegel, J.~F., F.~Kr{\"u}ger, A.~Jordan, and F.~Fasciati (2020).
\newblock Robust forecast evaluation of expected shortfall.
\newblock {\em Journal of Financial Econometrics\/}~{\em 18\/}(1), 95--120.

\end{thebibliography}

\newpage

\appendix

\makeatletter
\renewcommand \thesection{S\@arabic\c@section}
\renewcommand\thetable{S\@arabic\c@table}
\renewcommand \thefigure{S\@arabic\c@figure}
\renewcommand \theproposition{S\@arabic\c@proposition}
\makeatother

\begin{center}
\bf \Large    Supplement to ``Forecasting and Backtesting Gradient Allocations of Expected
Shortfall''
\end{center}

Section~\ref{sec:proof:details} provides proofs and details of some statements presented in the main paper.
Section~\ref{sec:details:empirical} gives a detailed description of the empirical analysis in Section~\ref{sec:empirical} and some additional results of backtesting ESCs in the one-step approach~\citep{fissler2024backtesting}.
Throughout the supplementary material, references starting with ``S'' refer to this supplement, and those without this prefix refer to the main paper.

\section{Proofs and details}\label{sec:proof:details}

\subsection{Order sensitivity and orientation}\label{sec:order:sensitivity}

A key property of scoring functions regarding the comparative backtests is \emph{order sensitivity}, which concerns whether two misspecified forecasts of ESCs are ordered by means of the multi-objective scoring function in Proposition~\ref{prop:elicitability:esc} if one tuple of forecasts is closer to the true ESCs than another.
The next proposition presents conditions under which two misspecified forecasts are ordered in this sense.

\begin{proposition}\label{prop:order:sensitivity}
Let $\bX \in \tilde{\mathcal F}_\text{c}^1(\mathbb{R}^d)$.
For $v\in \mathbb{R}$ and $j\in \{1,\dots,d\}$, denote by $\mathcal I_j(v)$ an open interval with endpoints $\operatorname{ESC}_\alpha(X_j,S)$ and $\mathbb{E}[X_j|S>v]$.
\begin{enumerate}
    \item[(O1)]
    Fix $j\in \{1,\dots,d\}$.
    Let $\mathbf{S}_j$ be a multi-objective scoring function in Proposition~\ref{prop:elicitability:esc} (S1). 
    For two misspecified forecasts $(m_j,v),(m_j^\ast,v^\ast)\in \mathbb{R}^2$ of $(\operatorname{ESC}_\alpha(X_j,S),\operatorname{VaR}_\alpha(S))$, we have that $\mathbb{E}[{\mathbf S}_j((m_j,v),\bX)]\le_{\text{lex}}  \mathbb{E}[{\mathbf S}_j((m_j^\ast,v^\ast),\bX)]$  if one of the following cases hold:
    \begin{enumerate}
        \item[(i)] $v^\ast<v\le \operatorname{VaR}_\alpha(S)$;
        \item[(ii)] $\operatorname{VaR}_\alpha(S)\le v< v^\ast$;
        \item[(iii)] $v=v^\ast$ and $m_j^\ast<m_j\le \tilde m$ for all $\tilde m\in \mathcal I_j(v)$;
        \item[(iv)] $v=v^\ast$ and $\tilde m\le m_j<m_j^\ast$ for all $\tilde m\in \mathcal I_j(v)$.
    \end{enumerate}
\item[(O2)]
Let $\mathbf{S}$ be a multi-objective scoring function in Proposition~\ref{prop:elicitability:esc} (S2). 
For two misspecified forecasts $(\bm{m},v),(\bm{m}^\ast,v^\ast)\in \mathbb{R}^{d+1}$ of $(\operatorname{ESC}_\alpha(\bX,S),\operatorname{VaR}_\alpha(S))$, we have that $\mathbb{E}[{\mathbf S}((\bm{m},v),\bX)]\le_{\text{lex}}  \mathbb{E}[{\mathbf S}((\bm{m}^\ast,v^\ast),\bX)]$ if one of the following cases hold: (i); (ii); (iii) for all $j\in\{1,\dots,d\}$; (iv) for all $j\in\{1,\dots,d\}$.
     \end{enumerate}
 \end{proposition}

\begin{proof}
    (O2) is a direct consequence from (O1), and thus we will show (O1) for a fixed $j\in \{1,\dots,d\}$.
It is shown in Section~2.4~ of \citet{ehm2016quantiles} that $\mathbb{E}[\S^\text{VaR}(v,S)]< \mathbb{E}[\S^\text{VaR}(v^\ast,S)]$ holds if (i) or (ii) holds. Therefore, each of (i) and (ii) implies $\mathbb{E}[{\mathbf S}_j((m_j,v),\bX)]\le_{\text{lex}}  \mathbb{E}[{\mathbf S}_j((m_j^\ast,v^\ast),\bX)]$.

We next suppose that (iii) holds.
In this case, we have that $\mathbb{E}[S^\text{VaR}(v,S)]= \mathbb{E}[S^\text{VaR}(v^\ast,S)]$.
For any $\eta\in\R$ and $\bx \in \mathbb{R}^d$, the elementary scoring function $\S_{j\eta}^\text{ESC}$ in Proposition~\ref{prop:esc:mixture} (M1) satisfies:
\begin{align*}
\S_{j,\eta}^\text{ESC}((m_j^\ast,v),\bx)- \S_{j,\eta}^\text{ESC}((m_j,v),\bx) = \id{\{s>v\}}(x_j-\eta)(\id{\{m_j^\ast\le \eta\}}-\id{\{m_j\le \eta\}}),
\end{align*}
and thus:
\begin{align*}
\mathbb{E}[\S_{j,\eta}^\text{ESC}((m_j^\ast,v),\bX)]&- \mathbb{E}[\S_{j,\eta}^\text{ESC}((m_j,v),\bX)] \\
&= \mathbb{P}(S>v)(\mathbb{E}[X_j|S>v]-\eta)(\id{\{m_j^\ast\le \eta\}}-\id{\{m_j\le \eta\}}),
\end{align*}
which is nonnegative for all $\eta \in \mathbb{R}$ and positive for $m_j^\ast <\eta <m_j$.
Since $H_j((m_j^\ast,m_j))>0$, the mixture representation in Proposition~\ref{prop:esc:mixture} (M1) yields:
\begin{align*}
    \mathbb{E}[\S_j^\text{ESC}((m_j,v),\bX)]< \mathbb{E}[\S_j^\text{ESC}((m_j^\ast,v^\ast),\bX)],
    \end{align*}
    which implies $\mathbb{E}[{\mathbf S}_j((m_j,v),\bX)]\le_{\text{lex}}  \mathbb{E}[{\mathbf S}_j((m_j^\ast,v^\ast),\bX)]$.
(iv) is shown analogously.
\end{proof}

To summarize it briefly, the average scores of the two misspecified forecasts of ESCs are ordered when (i) the corresponding total VaRs are ordered, or (ii) the misspecified total VaRs are equal (in this case,  $\S_1^{\text{ESC}},\dots,\S_d^{\text{ESC}}$ elicit \emph{biased} ESCs), the two misspecified ESCs are ordered in the component-wise sense, and they are not between the true and biased ESCs.
Regarding the latter case, misspecification of total VaR produces intervals such that forecasts of ESCs belonging there have inconsistent orders of average scores.

Although various concepts of order sensitivity are proposed in~\citet{fissler2019order}, we are not aware of any improvement on Proposition~\ref{prop:order:sensitivity}
except for some special cases; see Sections~3.3.1 and 3.3.2 of \citet{fissler2019order} for related discussions.

Regarding the concept of \emph{orientation}~\citep{steinwart2014elicitation}, it is useful to associate the sign of $\mathbf{\bar V}$ with over- and under-estimations of ESCs since under-estimation of risk functionals is considered to be more problematic than over-estimation from a regulatory viewpoint.
With the identification functions in Proposition~\ref{prop:joint:identifiability:esc}, we have, for $\bX \in \mathcal F$, that $\mathbb{E}[\V^{\text{VaR}}(v,S)]\le 0$ if and only if $\operatorname{VaR}_\alpha(S)\le v$; i.e., the forecast $v$ over-estimates the total VaR.
Moreover, under the correct specification $v=\operatorname{VaR}_\alpha(S)$ and for $j\in \{1,\dots,d\}$, a forecast $m_j \in \mathbb{R}$ of the $j$th ESC is over-estimated; i.e., $\operatorname{ESC}_\alpha(X_j,S)\le m_j$, if and only if $\mathbb{E}[\V_j^{\text{ESC}}((m_j,v),\bX)]\le 0$.

\subsection{Proof of Proposition~\ref{prop:esc:mixture}}
(M2) is an immediate consequence of (M1), and thus it suffices to show (M1).
We fix $j\in \{1,\dots,d\}$.
For a strictly convex function $\phi_j:\R\rightarrow \R$ with derivative $\phi_j'$, define the Bregman-type function $\Phi_j:\mathbb{R}\times \mathbb{R}\rightarrow\mathbb{R}$ as:
\begin{align*}
\Phi_j(m,x)=\phi_j(x)-\phi_j(m)-\phi_j'(m)(x-m).
\end{align*}
Following the proof of Theorem~1 in~\citet{ehm2016quantiles}, it holds that:
\begin{align*}
\Phi_j(m,x)=\int_{m}^{x}(x -\eta)\,\mathrm{d}\phi_j'(\eta),\quad\text{ for $m,x\in \mathbb{R}$ such that } m < x.
\end{align*}
Therefore, for $m\in \mathbb{R}$ and $\bx =(x_1,\dots,x_d)\in \mathbb{R}^d$ such that $m<x_j$, we have that:
\begin{align*}
\S_j^\text{ESC}((m,v),\bx)&=\id{\{s > v\}} \left\{
\phi_j'(m)(m-x_j) - \phi_j(m)+\phi_j(x_j)
\right\}\\
&= \id{\{s > v\}}\Phi_j(m,x_j)\\
&= \int_{\mathbb{R}} \S_j^\text{ESC}((m,v),\bx)\,\mathrm{d}\phi_j'(\eta).
\end{align*}
This representation is shown to hold analogously for the case $m>x_j$ and is trivial for the case $m=x_j$.

\subsection{Interpretation of the elementary scoring function}\label{sec:interpretation}

The elementary scoring function~\eqref{eq:esc:elementary} can be interpreted as a degree of regret in the following setting.
Suppose that the $j$th branch of a company has a fixed capital $\eta$ to cover a future loss, whose point forecast is denoted by $m_j$, incurred at this branch under the specific situation when the  company incurs a loss greater than $v$.
If $m_j \le \eta$, then the branch expects that the initial capital can cover a future loss.
In this case, since an excess loss $(x_j-\eta)$ is incurred when the realized loss $x_j$ exceeds $\eta$ and $s>v$ occurs, the amount $\S_{j,\eta}^\text{ESC}((m_j,v),\bx)=\id{\{s>v\}}(x_j-\eta)_{+}$ can be understood as a degree of regret against the initial expectation.
If $\eta<m_j$, then the branch expects that the initial capital does not cover a future loss, and thus some risk treatment can be conducted.
In this case, the branch realizes the squandered opportunity of capital reduction $(\eta-x_j)_{+}$ when $s>v$ occurs, and thus the degree of regret against the initial expectation can be measured by $\S_{j,\eta}^\text{ESC}((m_j,v),\bx)=\id{\{s>v\}}(\eta-x_j)_{+}$.
Provided $v=\operatorname{VaR}_\alpha(S)$, the true $j$th ESC is the optimal $m_j$ minimizing the expected degree of regret.

\section{Details of the empirical analysis}\label{sec:details:empirical}

In this section we describe details of the empirical analysis omitted in Section~\ref{sec:empirical}.
For brevity, we describe models for the case when $\mathbb{T}_{\text{in}}=\{1,\dots,n\}$ and $\mathbb{T}_{\text{out}}=\{n+1\}$, and thus $\mathbb{T}=\{1,\dots,n+1\}$.

\subsection{Bottom-up and top-down approaches}\label{sec:top:bottom}

We classify an approach of building dynamic models of ESCs into the bottom-up and top-down approaches.
In the \emph{bottom-up approach} we estimate the risk functional of $\bX_{n+1}
|\mathcal G_n$ based on a joint model $\{\bX_t\}_{t \in \mathbb{T}}$.
This joint model determines the dynamics of $\{S_t\}_{t \in \mathbb{T}}$ and $\{(X_{j,t},S_t)\}_{t \in \mathbb{T}}$, based on which we predict total VaR, total ES, and ESCs at time $n+1$.
An example is the bottom-up GARCH model described in Section~\ref{sec:detailed:description}.
A major advantage of this approach is that the estimated model can be used to estimate quantities other than ESCs.
This versatility, however, comes at the expense of possibly low forecast accuracy of total VaR and ES since the aggregate loss is modeled only indirectly.
Another potential drawback is the challenge of modeling $d$-dimensional time series especially when $d$ is large.

In the \emph{top-down approach} we first estimate the dynamics of total VaR on $\mathbb{T}$ and then estimate the dynamics of ESCs on $\mathbb{T}$ based on some necessary model specification.
For example, once the dynamics of total VaR is specified, that of the $j$th ESC can be estimated only from the bivariate time series $\{(X_{j,t},S_t)\}_{t \in \mathbb{T}}$ for each $j\in\{1,\dots,d\}$ without specifying the joint model $\{\bX_{t}\}_{t \in \mathbb{T}}$.
Total ES can be predicted together with total VaR in the first stage or estimated as the sum of the ESCs estimated in the second stage.
Compared with the bottom-up approach, the potential benefits of the top-down approach are flexibility of the VaR model and reduction of modeling effort to estimate ESCs.
On the other hand, this model specification can be involved in the \emph{compatibility problem} that concerns the existence of a joint model $(\bX_t,S_t)$ whose bivariate marginal time series $(X_{j,t},S_t)$, $j=1,\dots,d$, are the specified ones.
Complexity of this compatibility problem may hinder the use of the estimated model other than the forecasting problem of ESCs.

\subsection{Detailed model description}\label{sec:detailed:description}

We now describe details of the models considered in Section~\ref{sec:empirical}.

\noindent\emph{Historical simulation (HS)}:
In this model the estimator $\widehat{VaR}_{n+1}^\text{HS}$ is given as the empirical $\alpha$-quantile based on $S_{1},\dots,S_{n}$.
We then predict the total ES and ESCs as follows:
\begin{align*}
\widehat{ESC}_{j,n+1}^{\text{HS}}&=\frac{1}{n}\sum_{t=1}^n X_{j,t}\id{\left\{S_{t} > \widehat{VaR}_{n+1}^\text{HS}\right\}},\quad j=1,\dots,n,\\
\widehat{ES}_{n+1}^{\text{HS}}&=\sum_{j=1}^n \widehat{ESC}_{j,n+1}^{\text{HS}}.
\end{align*}
\emph{Bottom-up GARCH (GARCH.BU)}:
In this fully parametric model we assume a copula-GARCH model with skew-$t$ residuals on $\{\bX_{t}\}_{t\in\mathbb{T}}$.
For $j=1,\dots,d$, we assume an AR-GARCH(1,1) model $X_{j,t}=\mu_{j}+\sigma_{j,t}Z_{j,t}$ on the marginal time series $\{X_{j,t}\}_{t \in \mathbb{T}}$, where $\mu_{j}\in\mathbb{R}$ is a location parameter, and $\sigma_{j,t}>0$ is a conditional standard deviation following the GARCH(1,1) dynamics.
The residual series $\{Z_{j,t}\}_{t\in \mathbb{T}}$ is assumed to be a strict white noise of a skew-$t$ distribution $\operatorname{Skt}(\nu_j, \gamma_j)$~\citep{fernandez1998bayesian} for degrees of freedom $\nu_j>0$ and skewness parameter $\gamma_j > 0$.
For the dependence structure, we assume that $(Z_{1,t},\dots,Z_{d,t})$, $t\in \mathbb{T}$, has a $t$-copula with degrees of freedom $\nu>0$ and correlation matrix $P$.
Under this model assumption, the joint distribution of $\bX_{n+1}|{\cal G}_{n}$ is specified by the skew-$t$ margins and the $t$-copula.
We first estimate parameters of the marginal GARCH models by the maximum likelihood estimation using the package \textsf{fGarch}~\citep{fgarch} implemented in \textsf{R}~\citep{Rcore}.
We then estimate parameters of the $t$-copula from $(Z_{1,t},\dots,Z_{d,t})$, $t\in \mathbb{T}_{\text{in}}$, by the maximum pseudo-likelihood estimation with inversion of Kendall's tau~\citep{demarta2005t}, which is implemented in the pacakge \textsf{copula}~\citep{copula}.
Based on parameters' estimates we forecast total VaR, total ES, and ESCs of $\bX_{n+1}|{\cal G}_{n}$ by Monte Carlo simulation with sample size $10^5$.
\\
\noindent
\emph{Top-down GARCH (GARCH.TD)}:
In this model we assume a bivariate copula-GARCH model with skew-$t$ residuals on $\{(X_{j,t},S_t)\}_{t \in \mathbb{T}}$ for each $j\in \{1,\dots,d\}$, where the marginal AR-GARCH model on $\{S_t\}_{t \in \mathbb{T}}$ is shared in common.
Based on the AR-GARCH(1,1) assumption on $\{S_t\}_{t\in \mathbb{T}}$ with skew-$t$ residuals, we estimate the dynamics of total VaR and total ES parametrically.
We then predict ESCs of $\bX_{n+1}|{\cal G}_{n}$ through Monte Carlo simulation.
Note that the model of the vector of residuals $(Z_{1,t},\dots,Z_{d,t})$, $t \in \mathbb{T}$, is left unspecified since we fit a common AR-GARCH(1,1) model $S_t = \mu + \sigma_t Z_t$ on the series $\{S_t\}_{t\in\mathbb{T}}$ and directly model the copula of $(Z_{j,t},Z_t)$, $t\in \mathbb{T}$, for $j=1,\dots,d$.\\
\noindent
\emph{Hysteretic autoregressive model with GARCH error and dynamic conditional correlations (HAR.GARCH)}:
In this model we consider the bivariate hysteretic autoregressive (HAR) model with GARCH error and dynamic conditional correlations~\citep{tse2002,chen2019quantile} on the time series $\{(X_{j,t},S_t)\}_{t\in \mathbb{T}}$ for $j=1,\dots,d$.
We say that a pair of asset return series $\utwi{y}_t=(y_{1,t},y_{2,t})$ follows an HAR.GARCH model if:
\begin{eqnarray*}
\utwi{y}_{t} &=&  \utwi{\Phi}_0^{(k)}+ \utwi{\Phi}_1^{(k)} \utwi{y}_{t-1} + \utwi{a}_t \quad \mbox{if }\:\:
R_t=k, \quad k\in\{1,2\}, \nonumber\\
a_{i,t}&=&\sqrt{h_{i,t}}\epsilon_{i,t},\quad \utwi{\epsilon}
_{t}\sim D(0,\utwi{\Gamma} _{t}),\quad i=1,2,  \nonumber\\
h_{i,t} &=& \alpha_{i,0}^{(k)} + \alpha_{i,1}^{(k)}a_{i,t-1}^{2} +\beta _{i,1}^{(k)}h_{i,t-1},  \nonumber \\
\utwi{\Gamma}_{t} &=& \left( 1-\theta_{1}^{(k)} - \theta_{2}^{(k)}\right) \utwi{\Gamma}^{(k)} + \theta_{1}^{(k)}\utwi{\Gamma} _{t-1} +
\theta_{2}^{(k)}\utwi{\Psi}_{t-1},
\end{eqnarray*}
where $R_t$ is a regime indicator defined by:
\begin{align*}
R_t = \left\{ \begin{array}{ll}
1,        & \mbox{if } z_{t}<r_L, \\
 R_{t-1}, & \mbox{if } r_L \leq z_{t} \leq r_U, \\
       2, & \mbox{if } z_{t}>r_U,
\end{array} \right.
\end{align*}
with $z_t$ being a hysteresis variable, and $\utwi{\Psi}_{t-1}=(\psi_{uv,t-1})_{u,v\in\{1,2\}}$ is a local correlation matrix defined by:
\begin{align*}
\psi _{uv,t-1}=\frac{\sum\limits_{h=1}^{3}\epsilon _{u,t-h}\epsilon _{v,t-h}}{\sqrt{\left( \sum\limits_{h=1}^{3}\epsilon _{u,t-h}^{2}\right) \left(
\sum\limits_{h=1}^{3}\epsilon _{v,t-h}^{2}\right) }}.
\end{align*}

Parameters of the HAR.GARCH model include $\utwi{\Phi}_0^{(k)}\in\mathbb{R}^{2}$, $\utwi{\Phi}_1^{(k)}\in\mathbb{R}^{2\times 2}$, $r_L,r_U\in \mathbb{R}$, and a positive-definite matrix $\utwi{\Gamma}^{(k)}\in\mathbb{R}^{2\times 2}$ with unit diagonal elements and non-negative real numbers $\alpha_{i,0}^{(k)}$, $\alpha_{i,1}^{(k)}$, $\beta _{i,1}^{(k)}$, $\theta_1^{(k)}$, and $\theta_2^{(k)}$ such that $0<\alpha_{i,1}^{(k)}+\beta _{i,1}^{(k)}<1$ and $0 < \theta_1^{(k)} + \theta_2^{(k)} < 1$ for $i,k\in \{1,2\}$.
For a hysteresis variable, we utilize an endogenous variable $z_t = y_{1,t}$ to exhibit hysteresis effects. 
Following \citet{chen2019quantile}, we choose an adapted multivariate Student t distribution, based on the scale mixtures of the normal representation, as the distribution $D(0,\utwi{\Gamma} _{t})$ for the bivariate vector of residuals $\utwi{\epsilon}_{t}$.
We utilize Bayesian methods to estimate parameters of the above model and the risk quantities of interest; see Sections~3 and~4 of~\citet{chen2019quantile}, respectively.
In our setting, $y_{1,t}$ stands for $X_{j,t}$, $j\in \{1,\dots,d\}$, and $y_{2,t}$ represents $S_{t}$.
As the total VaR is predicted $d$-times, we simply use the average of them.
Moreover, total ES is predicted as the sum of ESCs.\\
\noindent
\emph{Compositional regression model with least square estimation (CR.LSE)}:
Following~\citet{boonen2019forecasting}, we first obtain a series of ESCs by the \emph{elliptical formula}~\citep[Corollary~8.43 of][]{mcneil2015quantitative}:
\begin{align*}   \widehat{ESC}_t^{\text{EL}}&=\left(\widehat{ESC}_{1,t}^{\text{EL}},\dots,\widehat{ESC}_{d,t}^{\text{EL}}\right)^\top\\
    &=
    \hat{\utwi{\mu}} + \frac{\hat{\Sigma}_t \bone_d}{\bone_d^\top\hat{\Sigma}_t \bone_d} \left(\widehat{ES}_t^{\text{EL}}- \bone_d^\top\hat{\utwi{\mu}}\right),\quad t \in \mathbb{T},
\end{align*}
where $\hat{\utwi{\mu}}=(\hat{\mu_1},\dots,\hat{\mu_d})$ is the vector of location parameters, and $\hat{\Sigma}_t$ is the conditional covariance matrix estimated in GARCH.BU.
For the dynamics of total VaR and total ES, we use the same estimates used in
GARCH.TD. We then obtain the allocation weights as $\hat{\bm{w}}_t = \operatorname{C}\left(\widehat{ESC}_t^{\text{EL}}\right)$, $t \in \mathbb{T}$.
Regarding this set of allocation weights as compositional data, we estimate the parameter $\utwi{\theta}$ of $\upsilon_{\utwi{\theta}}$ in ~\eqref{eq:w:model} as follows:
\begin{align*}
\hat{\utwi{\theta}}=\operatorname{argmin}_{\utwi{\theta}\in \Theta} \sum_{t=1}^{n-1}\sum_{j=1}^d\left\{\hat{w}_{j,t+1}-\upsilon_{\utwi{\theta}}(\hat{\bm{w}}_s,\bX_s, s\le t)_j\right\}^2.
\end{align*}
Finally, we use this $\hat{\utwi{\theta}}$ to predict ESCs of $\bX_{n+1}|\mathcal G_n$ following Section~\ref{sec:proposed:model}.
\\
\noindent
\emph{Compositional regression model based on score optimization (CR.OPT)}:
This is the proposed model described in Section~\ref{sec:proposed:model}.
For the dynamics of total VaR and total ES, we use the same estimates used in CR.LSE.
We choose $\bm{w}_1=\operatorname{C}\left(\widehat{ESC}_1^{\text{EL}}\right)$ for the initial allocation weight.
When optimizing~\eqref{eq:theta:score}, we use the vector of parameters estimated in CR.LSE as the initial value.

\subsection{Backtests in the one-step approach}\label{sec:one:step:approach}

Due to multi-objective elicitability of ESCs, we also conduct two-sided and ``one and a half-sided'' \emph{Wald-tests} for tuple of ESCs and each of them in combination with total VaR; see Section~5 of \citet{fissler2024backtesting} for details of these tests.
In this one-step approach the major difference from the two-step approach is that the null hypothesis includes equality in the accuracy of total VaRs for two competing forecasts .
We report p-values of these tests in Table~\ref{table:S:test:wald}.

\begin{table}[p]
\caption{Results of the Wald-tests to compare the forecast accuracy of ESCs with CR.OPT as the benchmark model.
}
\label{table:S:test:wald}
\centering
\scalebox{0.85}{
\begin{tabular}{lrrr r lrrr}
  \hline
    &  \multicolumn{3}{c}{p-value$^a$} & &     &  \multicolumn{3}{c}{p-value$^a$} \\  \hline
H0 & $=$ CR.OPT & $\le$ CR.OPT & $\ge$ CR.OPT && H0 & $=$ CR.OPT & $\le$ CR.OPT & $\ge$ CR.OPT  \\
  \hline\\
  \multicolumn{4}{l}{(1) Tuple of ESCs} & &      \multicolumn{4}{l}{(2) ESC (AMZN)}\\
HS & {\bf 0.006} & {\bf 0.012} & {\bf 0.004} &  & HS & {\bf 0.010} & {\bf 0.012} & {\bf 0.006} \\
  GARCH.BU & 0.054 & 0.404 & {\bf 0.035} & &  GARCH.BU & 0.533 & 0.398 & 0.404 \\
  GARCH.TD &0.524 & 1.000 & 0.390 & &  GARCH.TD & 0.963 & 1.000 & 0.874 \\
  HAR.GARCH & 0.418 & 0.929 & 0.302 &  &  HAR.GARCH & 0.739 & 0.588 & 0.929 \\
  CR.LSE & 0.228 & 1.000 & 0.157& &CR.LSE & 0.358 & 1.000 & 0.255 \\
  CR.OPT & ---- & ---- & ---- &  &CR.OPT & ---- &----  & ---- \\
  \\
      \multicolumn{4}{l}{(3) ESC (GOOGL)}& &      \multicolumn{4}{l}{(4) ESC (TSLA)}\\
  HS & {\bf 0.019} & {\bf 0.012} & {\bf 0.012} &  & HS & {\bf 0.009} & {\bf 0.012} & {\bf 0.006} \\
  GARCH.BU & 0.434 & 0.404 & 0.315 &   &GARCH.BU & {\bf 0.034} & 0.404 & {\bf 0.021} \\
  GARCH.TD & 0.999 & 1.000 & 0.981 &   &  GARCH.TD & 0.526 & 1.000 & 0.392 \\
  HAR.GARCH & 0.559 & 0.929 & 0.420 & &   HAR.GARCH & 0.303 & 0.929 & 0.213 \\
  CR.LSE & 0.907 & 0.782 & 1.000 & &  CR.LSE & 0.211 & 1.000 & 0.144 \\
  CR.OPT & ---- & ---- & ---- & &CR.OPT & ---- & ---- & ---- \\
   \hline
\end{tabular}
}
\begin{tablenotes}
\item $^a$The p-values are calculated based on the three different null hypotheses with all of them including the equality of the accuracy of forecasted total VaRs, where ``$=$ CR.OPT" means that the model is equally accurate as CR.OPT, and ``$\le$ ($\ge$) CR.OPT" represents the hypothesis that the model is less (more) accurate than CR.OPT.
\end{tablenotes}
\end{table}

Due to joint identifiability of total ES~\citep[see~Section~2.1~of~][]{nolde2017elicitability} and $j$th ESC for $j=1,\dots,d$
 (see Proposition~\ref{prop:joint:identifiability:esc}) in combination with total VaR, we also conduct Wald-tests for these risk quantities in the one-step approach.
Namely, for total ES, we consider the two-sided Wald-test for the null hypothesis:
\begin{align*}
H_0^{=}: \mathbb{E}\left[\bar \V_T^{\text{VaR}}\right]=0 \quad \text{and}\quad  \mathbb{E}\left[\bar \V_T^{\text{ES}}\right] = 0,
\end{align*}
and the ``one and a half-sided'' Wald-test for each of the null hypotheses:
\begin{align*}
H_0^{\le_\text{lex}}&:    \mathbb{E}\left[\bar \V_T^{\text{VaR}}\right] =0\quad\text{and}\quad
        \mathbb{E}\left[\bar \V_T^{\text{ES}}\right] \le 0,\\
H_0^{\ge_\text{lex}}&:    \mathbb{E}\left[\bar \V_T^{\text{VaR}}\right] =0\quad\text{and}\quad
        \mathbb{E}\left[\bar \V_T^{\text{ES}}\right] \ge 0,
\end{align*}
based on Section~5 of~\citet{fissler2024backtesting} as considered in the above comparative backtests.
For the $j$th ESC for $j\in\{1,\dots,d\}$, we replace $\bar \V_T^{\text{ES}}=(1/T)\sum_{t=1}^T \V_t^{\text{ES}}$ above with $\bar \V_{j,T}^{\text{ESC}}=(1/T)\sum_{t=1}^T \V_{j,t}^{\text{ESC}}$.
We consider these null hypotheses since over- and under-estimations of ESCs can be judged when the corresponding total VaR is calibrated; see Section~\ref{sec:order:sensitivity}.
We summarize p-values of these tests in Table~\ref{table:V:test:wald}.

\begin{table}[p]
\caption{
Results of the Wald-tests to verify forecast accuracy of total ES and ESCs.
}
\label{table:V:test:wald}
\centering
\scalebox{1}{
\begin{tabular}{lrrr r lrrr}
  \hline
      &  \multicolumn{3}{c}{p-value$^a$} &  &   &  \multicolumn{3}{c}{p-value$^a$} \\  \hline
H0 & $=$ True & $\le$ True & $\ge$ True &  & H0 & $=$ True & $\le$ True & $\ge$ True  \\
  \hline\\
    \multicolumn{4}{l}{(1) Total ES} & &     \multicolumn{4}{l}{(2) ESC (AMZN)}\\
HS & {\bf 0.000} & {\bf 0.000} & {\bf 0.000}  & &   HS & {\bf 0.000} & {\bf 0.000} & {\bf 0.000} \\
  GARCH.BU & {\bf 0.008} & {\bf 0.005} & {\bf 0.025} & &   GARCH.BU &  {\bf 0.038} & {\bf 0.024} & {\bf 0.025} \\
  GARCH.TD & 0.174 & 0.118 & 0.204  & &    GARCH.TD & 0.270 & 0.188 & 0.204 \\
  HAR.GARCH & 0.175 & 0.118 & 0.118  & &   HAR.GARCH & {\bf 0.016} & 0.118 & {\bf 0.010} \\
  CR.LSE & 0.174 & 0.118 & 0.204  & &   CR.LSE & 0.276 & 0.192 & 0.204 \\
  CR.OPT & 0.174 & 0.118 & 0.204  & &   CR.OPT &  0.215 & 0.147 & 0.204 \\ \\
     \multicolumn{4}{l}{(3) ESC (GOOGL)}& &       \multicolumn{4}{l}{(4) ESC (TSLA)}\\
  HS & {\bf 0.000} & {\bf 0.000} & {\bf 0.000} & &   HS & {\bf 0.000} & {\bf 0.000} & {\bf 0.000} \\
  GARCH.BU &{\bf 0.037} & {\bf 0.024} & {\bf 0.025} & &   GARCH.BU & {\bf 0.011} & {\bf 0.007} & {\bf 0.025} \\
  GARCH.TD & 0.238 & 0.164 & 0.204 & &   GARCH.TD &  0.230 & 0.158 & 0.204 \\
  HAR.GARCH & {\bf 0.003} & 0.118 & {\bf 0.002} & &   HAR.GARCH & {\bf 0.003} & 0.118 & {\bf 0.002} \\
  CR.LSE & 0.264 & 0.184 & 0.204 & &   CR.LSE & 0.262 & 0.182 & 0.204 \\
  CR.OPT & 0.275 & 0.191 & 0.204 & &   CR.OPT & 0.280 & 0.195 & 0.204 \\
   \hline
\end{tabular}
}
\begin{tablenotes}
    \item $^a$The p-values are calculated based on the three different null hypotheses with all of them including the preciseness of forecasted total VaRs, where ``$=$ True'' means that the forecast is precise, ``$\le$ True'' represents the hypothesis that the forecast is under-estimated, and ``$\ge$ True'' stands for the case when the forecast is over-estimated.
\end{tablenotes}
\end{table}

\end{document}